\documentclass[prl,twocolumn,lengthcheck]{revtex4-1}

\usepackage{amsmath,amsfonts,amssymb,amsthm}
\usepackage{graphicx,svg}
\usepackage{hyperref}
\usepackage{color}
\usepackage{braket}
\usepackage{mathrsfs}
\usepackage{tikz}
\usepackage{thmtools}
\usepackage{mathtools}
\usepackage{todonotes}
\usepackage{braids}

\usetikzlibrary{arrows, calc,decorations.markings,cd,external}

\DeclareMathOperator{\mult}{mult}

\newcommand{\figref}[1]{Fig.~\ref{#1}}
\newcommand{\markedPoint}{\star}
\newcommand{\rotatedmultiplication}[1]{\stackrel{#1}{\curvearrowleft}\,}

\DeclareMathOperator{\rot}{rot}
\newcommand{\morphism}[1]{\xrightarrow{#1}}
\newcommand{\catname}[1]{{\normalfont\textbf{#1}}}
\newcommand{\inv}{^{-1}}

\DeclarePairedDelimiter{\floor}{\lfloor}{\rfloor}

\DeclareMathOperator{\id}{id}

\theoremstyle{plain}
\newtheorem{theorem}{Theorem}
\newtheorem{lemma}[theorem]{Lemma}
\newtheorem{proposition}[theorem]{Proposition}

\theoremstyle{definition}
\newtheorem{definition}{Definition}

\newtheorem{example}{Example}
\newtheorem{fact}{Fact}

\theoremstyle{remark}

\begin{document}

\title{Perfect tangles}

\author{Johannes Berger$^{1}$ and Tobias J.\ Osborne$^{2}$}
\affiliation{$^{1}$Fachbereich Mathematik, Universit\"at Hamburg, Bundesstr.\ 55, 20146 Hamburg, Germany}
\affiliation{$^{2}$Institut f\"ur Theoretische Physik, Leibniz Universit\"at Hannover, Appelstr.\ 2, 30167 Hannover, Germany}

\begin{abstract}
We introduce  \emph{perfect tangles} for modular tensor categories. These are intended to generalise the \emph{perfect tensors} first introduced in the context of a toy model for the AdS/CFT correspondence. We construct perfect tangles for several categories of relevance for topological quantum computation, including the Temperley-Lieb, Fibonacci, Kuperberg spider, and Haagerup planar algebras. A general inductive construction proposed by Vaughan Jones for perfect tangles is also described. 
\end{abstract}

\maketitle

\emph{Holographic duality}, especially in the form of the \emph{AdS/CFT correspondence}, is a profound insight in the study of quantum gravity \cite{maldacena_large-n_1999,maldacena_large_1998}. A tremendous number of advances have resulted from exploring the AdS/CFT correspondence, beginning \cite{gubser_gauge_1998,witten_anti_1998} with the original duality and continuing to the present day. There is now confirmation of the duality for a dizzying zoo of settings (see, e.g., \cite{hubeny_ads/cft_2015, harlow_tasi_2018} for recent related reviews). 

Now that the AdS/CFT correspondence is well understood at the level of fields, the search for microscopic realisations has commenced in earnest. A key breakthrough in this endeavour came in the paper \cite{pastawski_holographic_2015} of Pastawski, Harlow, Yoshida, and Preskill on \emph{holographic codes} where, partly inspired by insights \cite{almheiri_bulk_2015} of Almheiri, Dong, and Harlow, ideas from quantum information theory were exploited to build a kinematical microscopic model of the $\text{AdS}_{3}/\text{CFT}_{1+1}$ correspondence. This model was recently \cite{osborne_dynamics_2017} augmented to include dynamics and promises to give a concrete microscopic realisation of a (discrete) quantum gravity. The holographic code has attracted considerable interest from high energy and quantum information theorists alike, serving as a kind of rosetta stone for new progress at the crossroads between quantum gravity and quantum information research.  

The key building block for the holographic code is an object known as a \emph{perfect tensor}, which is an isometry between any bipartition of its legs \cite{raissi_constructing_2017,li_invariant_2016,peach_tensor_2017,donnelly_living_2016}. These highly nongeneric structures enjoy a variety of interesting features and are interesting in their own right, especially due to their equivalence to \emph{absolutely maximally entangled} (AME) states \cite{gisin_bell_1998,higuchi_how_2000,scott_multipartite_2004,brown_searching_2005,borras_multiqubit_2007,facchi_maximally_2008,facchi_classical_2010,gour_all_2010,zha_maximally_2012,helwig_absolute_2012,arnaud_exploring_2013,helwig_absolutely_2013,helwig_absolutely_2013,goyeneche_absolutely_2015,enriquez_maximally_2016,chen_graph_2016,bernal_existence_2017,huber_absolutely_2017} and \emph{quantum error correcting codes} (QECC) \cite{nielsen_quantum_2000}. Perfect tensors have been generalised in a number of ways to build a variety of new microscopic models of the AdS/CFT correspondence \cite{pastawski_code_2016,hayden_holographic_2016,bao_consistency_2015,yang_bidirectional_2016,bhattacharyya_exploring_2016,may_tensor_2016}. 

The original formulation of the perfect tensor is clearly crucial for holographic codes using tessellations of arbitrary dimensional hyperbolic manifolds. However, it is in two spatial dimensions -- the original formulation of the holographic code -- where the model acquires the greatest richness. This is due to a variety of reasons, not least of which is the existence of nontrivial quasiparticles known as \emph{anyons} \cite{pachos_introduction_2012,wang_topological_2010,wilczek_fractional_1990}. In this \emph{planar} setting it is reasonable to slightly relax the definition of a perfect tensor so that it is only required to be an isometry with respect to \emph{planar bipartitions} of the legs. It may be argued that the pleasing properties of the holographic code construction are preserved and in this case we call the resulting object a \emph{perfect tangle}.

In this Letter we argue that perfect tangles exist in a wider class of categories than just that of quantum spin systems. We explicitly show how to construct perfect tangles for categories of central importance for topological quantum computation, particularly, the Temperley-Lieb algebra and the Fibonacci category. Exact solutions are also provided for several other cubic categories. We conclude with an inductive procedure to build perfect tangles from smaller ones.

\emph{The Fibonacci category and the Temperley-Lieb algebra.---}We exemplify our constructions throughout in terms of a model of anyons known as \emph{Fibonacci anyons}, which are understood \cite{read_beyond_1999} to model the quasiparticles for a fractional quantum Hall liquid at filling fraction $\nu = 2+2/5$ \cite{xia_electron_2004}. The mathematical structure appropriate for the kinematical description of $N$ such Fibonacci anyons is a \emph{modular tensor category} (MTC) \cite{wang_topological_2010} known as the \emph{Fibonacci} or \emph{golden} category corresponding to $\textsl{SO}(3)_q$ at $1+ q^2+q^{-2}=\frac{1+\sqrt{5}}{2}$. This theory is described by two particle types, traditionally denoted $1$ and $\tau$, with the fusion rule $\tau\times \tau = 1+\tau$. The Hilbert space $\mathcal{H}_N$ of $N$ $\tau$-anyons is spanned by all possible fusion histories. There is a convenient orthonormal basis for $\mathcal{H}_N$ given by \emph{fusion paths}: 
\begin{equation}
\begin{tikzpicture}[scale=1,baseline=-1ex]
	\draw (-3.5,0) -- (3.5,0);
	\foreach \x in {-2.5,-1.5,-0.5,0.5,1.5,2.5}{
		\draw (\x, 0.0) -- (\x,0.8) node[pos=1.2] {$\tau$};
	}
	\foreach \x in {1,...,4}{
		\draw (\x - 3, -0.2) node {$x_{\x}$};
	}
	\draw (2,-0.2) node {$\cdots$};
	\draw (-3.7,0) node {$\tau$};
	\draw (3.7,0) node {$\tau$};
\end{tikzpicture},
\end{equation}
where $x_j \in \{1,\tau\}$, $j=1,2,\ldots, N-3$. In the category formulation we call the linear space spanned by fusion histories starting with $j$ $\tau$-anyons and ending with $k$ $\tau$-anyons the \emph{morphism space} $\text{Hom}(\tau^{\otimes j},\tau^{\otimes k})$. In this way $\mathcal{H}_N \equiv \text{Hom}(1,\tau^{\otimes N})$. 

The Fibonacci category contains more structure because it is an example of a \emph{planar algebra} (see the Supplementary Material for further details) where diagrams can be multiplied by stacking: whenever we encounter a loop we remove it and replace it with a factor $d$. The Fibonacci category contains a subalgebra spanned by all noncrossing planar diagrams known as the \emph{Temperley-Lieb algebra} $\textsl{TL}(q)$. We denote by $\textsl{TL}_{n}(q)$ the vector space of all noncrossing planar diagrams  with $n$ incoming and $n$ outgoing legs \footnote{For simplicity we will often omit the $q$ argument.}.

\emph{Perfect and planar perfect tensors.---}Before we get to the definition of a planar tangle it is helpful to review the perfect tensor definition introduced in \cite{pastawski_holographic_2015}. Recall that an \emph{$n$-leg tensor} is a set of (complex) numbers indexed by $n$ copies of an indexing set $J$. The \emph{range} of an index is the number $d=\lvert J \rvert$ \footnote{In the present setting this is sufficient. More generally, one could require the indices for the legs to be indexed by $n$ different indexing sets, but this is not necessary here. The range of the $j$th index would then be the size of the $j$th indexing set.}, and we commonly call the indices \emph{legs}. In particular, any $n$-leg tensor whose legs range from $1$ to $d$ can be seen as a vector in $\mathbb{C}^{d^n}$. But we can also interpret such a tensor as the coefficients of a linear map $T$ between complex vector spaces. A \emph{perfect tensor} is now an $n$-leg tensor $C=\{C_{j_1\cdots j_n}\}$, with indices ranging from $1$ to $d$, such that for any bipartition $A\cup A^c$ of $\{1,\ldots, n\}$ with $\lvert A \rvert\leq \lvert A^c \rvert$ the linear map
\begin{align*}
C^{A\to A^c} : \mathcal{H}_A \rightarrow\mathcal{H}_{A^c}
\end{align*}
defined by (sums implied)
\begin{align*}
C^{A\to A^c}\equiv C_{j_1\cdots j_n} \ket{j_i}_{i\in A^c}\bra{j_l}_{l\in A}
\end{align*}
is proportional to an isometry, or, equivalently, 
\begin{align*}
\left(C^{A\to A^c}\right)^\dagger C^{A\to A^c} \propto \mathbf{1}_{\mathcal{H}_A},
\end{align*}
where $\mathcal{H}_A \equiv \bigotimes_{j\in A} \mathbb{C}^d$. Perfect tensors are highly overdetermined, and it is remarkable that nontrivial solutions exist at all. At the present time there are now a wide variety of examples \cite{raissi_constructing_2017,li_invariant_2016,peach_tensor_2017,donnelly_living_2016}. The first are given by generalising the connection to quantum error correcting codes, as explained in \cite{pastawski_holographic_2015}. A second class of examples follows by exploiting the equivalence of perfect tensors with \emph{absolutely maximally entangled} states \cite{gisin_bell_1998,higuchi_how_2000,scott_multipartite_2004,brown_searching_2005,borras_multiqubit_2007,facchi_maximally_2008,facchi_classical_2010,gour_all_2010,zha_maximally_2012,helwig_absolute_2012,arnaud_exploring_2013,helwig_absolutely_2013,helwig_absolutely_2013,goyeneche_absolutely_2015,enriquez_maximally_2016,chen_graph_2016,bernal_existence_2017,huber_absolutely_2017}.  

Our focus, in this Letter, is on a wider class of objects than perfect tensors which are only required to be perfect with respect to \emph{planar} or \emph{connected bipartitions} \footnote{Our notion of planar perfect coincides with \emph{block perfect tensors} independently introduced by G.\ Brennen, R.\ Harris, and T.\ Stace.}: given a set $I=\{1,\ldots, n\}$, a connected bipartition $A\cup A^c$ of $I$ is such that $A$ consists of subsequent numbers, modulo $n$. For example, if $n=3$ then this is true for every bipartition, but if $n=4$, then the bipartition $\{1,3\}\cup\{2,4\}$ is not connected, while $\{1,4\}\cup\{2,3\}$ surely is.
An $n$-leg tensor $T=\{T_{j_1\cdots j_n}\}_{j_k=1}^d$ is \emph{planar perfect} if for any connected bipartition of its legs, say $A\cup A^c$, with $\lvert A \rvert\leq \floor{\frac{n}{2}}$ the associated map
$T^{A\rightarrow A^c}$ is proportional to an isometry. Any perfect tensor is also planar perfect. But the converse is false in general. An example, due to G.\ Brennen, R.\ Harris, and T.\ Stace, of a planar perfect tensor is given by the Steane code \cite{steane_multiple-particle_1996}. 

\emph{Perfect tangles for the Temperley-Lieb algebra.---}Before presenting the full conditions for a perfect tangle we first write out our first nontrivial example for the Temperley-Lieb algebra $\textsl{TL}_2(q)$ with $0<q\leq 2$:
\begin{equation}
	T \equiv \alpha\,  \tikz[scale=0.5, baseline=-1mm]{\foreach \x in {0, 0.5} \draw (\x, -0.5) -- (\x, 0.5);}
 \,+\beta\, \tikz[scale=0.5, baseline=-1mm]{\draw (0, 0.5) arc[start angle=-180, end angle=0, radius=2.5mm]; \draw (0, -0.5) arc[start angle=180, end angle=0, radius=2.5mm];}\,,
\end{equation}
where the perfect condition requires $|\alpha|=|\beta|$, $\alpha\overline{\beta} + \overline{\alpha}\beta + q|\alpha|^2 = 0$, and $\alpha\overline{\beta} + \overline{\alpha}\beta + q|\beta|^2 = 0$. Choosing normalisations, we can set $\alpha = 1$ and $\beta = e^{\pm i f(q)}$, where $f(q) = \arccos(-q/2)$. One can check that $T$ so defined obeys the equation $TT^\dag = \mathbb{I}$ if one multiplies vertically by stacking. Multiplying sideways also yields the (horizontal) identity. Finally, with respect to, e.g., the bipartition $\{1\}\cup\{234\}$ we see that $T$ is proportional to an isometry (the proportionality factor is $\sqrt{q}$). This tangle is an example of a \emph{biunitary connection} \cite{goodman_coxeter_1989,evans_takesaki_1989,jones_planar_1999,jones_classification_2014,morrison_little_2014,reutter_biunitary_2016}, which plays an important role in the construction of subfactors and which, remarkably, also obeys the \emph{Yang-Baxter equation} \cite{goodman_coxeter_1989}. It is this latter observation that provides the key insight into all our constructions.

Now that we have an example at hand we present the conditions a perfect tangle must satisfy. To do this we introduce the \emph{one-click rotation} operation $\text{rot}:\textsl{TL}_{k}(q)\rightarrow \textsl{TL}_{k}(q)$ which rotates anticlockwise the strands of a tangle by one step. This operation is illustrated here for $k=3$:
\begin{equation}
\text{rot}\left(\ \begin{tikzpicture}[scale=0.7,baseline=-1ex]
	\draw (-0.5,0.5) rectangle (0.5,-0.5) node[pos=.5] {$T$}; 
	\foreach \x in {0,0.4}
		\draw (\x, 0.5) -- (\x,0.8);
	\draw (-0.4,-0.5) -- (-0.4,-0.8);
	\draw (-0.4,0.5)  -- (-0.4,0.8);
	\draw (0.4,-0.5) -- (0.4,-0.8);
	\draw (0,-0.5) -- (0,-0.8);				
\end{tikzpicture}\ \right) = 
\begin{tikzpicture}[scale=0.7,baseline=-1ex]
	\draw (-0.5,0.5) rectangle (0.5,-0.5) node[pos=.5] {$T$}; 
	\foreach \x in {0,0.4}
		\draw (\x, 0.5) -- (\x,0.8);
	\draw (-0.4,-0.5) -- (-0.4,-0.8);
	\draw (-0.4,0.5) arc [start angle=0, end angle=180, radius=2mm] -- (-0.8,-0.8);
	\draw (0.4,-0.5) arc [start angle=180, end angle=360, radius=2mm] -- (0.8,0.8);
	\draw (0,-0.5) -- (0,-0.8);				
\end{tikzpicture}\ .
\end{equation}
Using rotation we introduce the \emph{rotated multiplication} operation for two vectors or \emph{tangles} $T$ and $S$, written infix, as
\begin{equation}
	T \rotatedmultiplication{n} S = \rot^n(T)\cdot \rot^{-n}(S),
\end{equation}
where $\cdot$ denotes multiplication via stacking and a negative exponent indicates \emph{clockwise rotation}. With the definition of rotated multiplication in hand we define a \emph{perfect tangle} $T$ to be any vector $T\in \textsl{TL}_k(q)$ which satisfies
\begin{equation}
	T \rotatedmultiplication{n} T^\dagger \propto T^\dagger \stackrel{n}{\curvearrowright} T \propto \mathbf{1}_k 
\end{equation}
for all $0\le n\le k$, where $\mathbf{1}_k$ is the identity (the multiplicative unit) and $\stackrel{n}{\curvearrowright}$ denotes the rotated multiplication with the first argument rotated clockwise and the second anticlockwise. These conditions are illustrated in Fig.~\ref{fig:PerfectTangleCondition}.

\begin{figure}[!htp]\centering
	\begin{tikzpicture}[scale=1.3]
		\begin{scope}
			\coordinate (center) at (0,0);
			\draw ($(center) + (-0.6,0.8)$) rectangle ($(center) + (0.6,0.3)$);
			\draw ($(center) + (-0.6,-0.8)$) rectangle ($(center) + (0.6,-0.3)$);

			\node (T) at ($(center) + (0,0.55)$) {$T$};
			\node (TStar) at ($(center) + (0,-0.55)$) {$T^\dagger$};
			
			\draw ($(center) + (-0.5,0.8)$) arc [start angle = 0, end angle= 180, radius=1.5mm] -- node[midway, left] {$n$} 
					($(center) + (-0.8,-0.8)$) arc [start angle = 180, end angle=360, radius=1.5mm];
			\draw ($(center) + (0.5,0.3)$) arc [start angle = 180, end angle=360, radius=1.5mm] --
					($(center) + (0.8,1.2)$);
			\draw ($(center) + (0.5,-0.3)$) arc [start angle = 180, end angle=0, radius=1.5mm] --
					($(center) + (0.8,-1.2)$);
					
			\draw ($(center) + (-0.3,0.8)$) -- ($(center) + (-0.3,1.2)$) node [midway, right] {\hspace*{-1mm}\footnotesize $k-n$};
			\draw ($(center) + (-0.3,-0.8)$) -- ($(center) + (-0.3,-1.2)$) node [midway, right] {\hspace*{-1mm}\footnotesize $k-n$};
			\draw ($(center) + (-0.3,-0.3)$) -- ($(center) + (-0.3,0.3)$) node [midway, right] {\hspace*{-1mm}\footnotesize $k-n$};
		\end{scope}
		\node (propto) at (1.5,0) {$\propto$};
		\begin{scope}
			\coordinate (center) at (3,0);
			\draw ($(center) + (-0.6,0.8)$) rectangle ($(center) + (0.6,0.3)$);
			\draw ($(center) + (-0.6,-0.8)$) rectangle ($(center) + (0.6,-0.3)$);

			\node (T) at ($(center) + (0,0.55)$) {$T^\dagger$};
			\node (TStar) at ($(center) + (0,-0.55)$) {$T$};
			
			\draw ($(center) + (0.5,0.8)$) arc [start angle = 180, end angle= 0, radius=1.5mm] -- 
					($(center) + (0.8,-0.8)$) arc [start angle = 360, end angle=180, radius=1.5mm];
			\draw ($(center) + (-0.5,0.3)$) arc [start angle = 360, end angle=180, radius=1.5mm] -- node[midway, left] {$n$} 
					($(center) + (-0.8,1.2)$);
			\draw ($(center) + (-0.5,-0.3)$) arc [start angle = 0, end angle=180, radius=1.5mm] --
					($(center) + (-0.8,-1.2)$);
					
			\draw ($(center) + (-0.3,0.8)$) -- ($(center) + (-0.3,1.2)$) node [midway, right] {\hspace*{-1mm}\footnotesize $k-n$};
			\draw ($(center) + (-0.3,-0.8)$) -- ($(center) + (-0.3,-1.2)$) node [midway, right] {};
			\draw ($(center) + (-0.3,-0.3)$) -- ($(center) + (-0.3,0.3)$) node [midway, right] {};
		\end{scope}
		\node (propTO) at (4.5,0) {$\propto$};
		\begin{scope}
			\coordinate (center) at (5,0);
			\draw  ($ (center) + (0,-1.2)$)  -- ($ (center) + (0,1.2)$) node [midway, right] (TextNode) {\hspace*{-1mm}\footnotesize $k$};
		\end{scope}
	\end{tikzpicture}
	\caption[Definition of perfect planar tangle]{The defining conditions for a $k$-tangle $T$ to be perfect. A number $n$ next to a string is shorthand for $n$ parallel strings.}
	\label{fig:PerfectTangleCondition}
\end{figure}
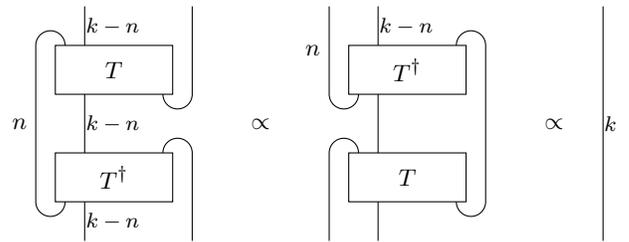
The number of equations obtained from computing the products of all $2n$ rotations of a general Temperley-Lieb $n$-tangle is rather frightening. Indeed, if the coefficients are assumed to be complex, then the number of (real) variables is $2$ times the dimension, i.e.\
\begin{equation*}
2 C_n = 2 \cdot \frac{1}{n+1}\binom{2n}{n},
\end{equation*}
where $C_n$ are the \emph{Catalan numbers}. The first few terms are $2,4,10,28,84,264,\ldots$, and the maximum number of defining equations for a perfect $n$-tangle is then $2nC_n$, which grows even more rapidly. As for the planar perfect tensor case, this system of equations is overdetermined. 

Luckily, we can reduce this a bit. In $\textsl{TL}_n$ there are no extra relations for the tangles, i.e., it is a \emph{free} vector space. It is thus necessary that a perfect tangle has a nonzero component, or \emph{support}, on $\mathbf{1}_n$, which allows us to normalize the coefficient of $\mathbf{1}_n$ so it has modulus $1$. For the sake of simplicity, and without loss of generality, we set $\alpha=1$. This already reduces the number of coefficients by two \footnote{But beware: This does not work in other situations, like trivalent categories for example.}.

We could also only look, say, at rotation-invariant tangles, i.e.\ $n$-tangles $T$ with $\rot(T)=T$. We then only have to compute $T\cdot T^\dagger$ and find solutions to get perfect tangles. Another possibility is to restrict to self-adjoint tangles, that is, tangles $T$ with $T^\dagger = T$. At the very least the generators of $\textsl{TL}_n$ (of which there are $n-1$) are self-adjoint, so that their respective coefficients have to be real. 

One can, for small values of $n$, numerically solve the defining equations and locate approximate solutions. Such approximate solutions provide initial guesses for exact solutions. In this way we have been able to build examples for self-adjoint examples in $\textsl{TL}_3$:
\begin{equation*}
T=
\begin{tikzpicture}[scale = 0.7, baseline=-1mm]
		\foreach \x in {-0.5,0,0.5}
			\draw (\x, 0.8) -- (\x,-0.8);
	\end{tikzpicture}
	+ \alpha\,
		\begin{tikzpicture}[scale = 0.7, baseline=-1mm]
			\draw (0,0.8) .. controls (0,-0.5+0.8) and (0.5,-0.5+0.8) .. (0.5,0.8);
			\draw (-0.5,-0.8) .. controls (-0.5,1-0.5-0.8) and (0,1-0.5-0.8) .. (0,-0.8);
			\draw (-0.5,0.8) .. controls (-0.5,0) and (0.5,0) .. (0.5, -0.8);
		\end{tikzpicture}
	+\overline{\alpha}\,
		\begin{tikzpicture}[scale = 0.7, baseline=-1mm]
			\draw (-0.5,0.8) .. controls (-0.5,-0.5+0.8) and (0,-0.5+0.8) .. (0,0.8);
			\draw (-0,-0.8) .. controls (-0,1-0.5-0.8) and (0.5,1-0.5-0.8) .. (0.5,-0.8);
			\draw (-0.5,-0.8) .. controls (-0.5,0) and (0.5,0) .. (0.5, 0.8);
		\end{tikzpicture}\,
+\beta\,
	\begin{tikzpicture}[scale = 0.7, baseline=-1mm]
		\draw (-0.5,0.8) -- (-0.5,-0.8);
		\draw (0,0.8) .. controls (0,-0.5+0.8) and (0.5,-0.5+0.8) .. (0.5,0.8);
		\draw (-0,-0.8) .. controls (-0,1-0.5-0.8) and (0.5,1-0.5-0.8) .. (0.5,-0.8);
	\end{tikzpicture}
	\,+ \,\gamma\,
	\begin{tikzpicture}[scale = 0.7, baseline=-1mm]
		\draw (0.5,0.8) -- (0.5,-0.8);
		\draw (-0.5,0.8) .. controls (-0.5,-0.5+0.8) and (0,-0.5+0.8) .. (0,0.8);
		\draw (-0.5,-0.8) .. controls (-0.5,1-0.5-0.8) and (0,1-0.5-0.8) .. (0,-0.8);
	\end{tikzpicture}\,,
\end{equation*}
where $\beta$ and $\gamma$ are necessarily real. For $q=2$, the only solution is given by
\begin{align*}
\alpha = 1,\qquad
\beta= -1,\qquad
\gamma = -1.
\end{align*}
If now $q=2\cos\frac{\pi}{n}$ for $n\geq 3$, then we define
\begin{align*}
f_{SA}(q) \equiv \frac{\pm \sqrt{4 - q^2} - 2}{q}.
\end{align*}
With this, the self-adjoint solutions are given by
\begin{align*}
\alpha = 1,\qquad
\beta= f_{SA}(q),\qquad
\gamma = \frac{1}{f_{SA}(q)}.
\end{align*}
Finally if $n=6$, that is $q=\sqrt{3}$, another solution given by
\begin{align*}
\alpha = 1,\qquad
\beta= -q,\qquad
\gamma = -q
\end{align*}
exists.

In the case of rotation-invariant solutions we have $\alpha = 1$, and $\beta = \gamma$. In fact, we have already seen one solution above: If $q=2$, then the self-adjoint tangle is also clearly rotationally invariant. If $q=2\cos\frac{\pi}{n}$ for $n\geq 6$ then we have solutions described by
\begin{align*}
\beta = 
-\frac{q}{q^2-2}
\pm i
\sqrt{-\frac{q^4-7 q^2+12}{\left(q^2-2\right)^2}}.
\end{align*}
These are in fact valid whenever $\sqrt{3}\leq q < 2$.

As an aside, we also found that rotation-invariant perfect tangles in $\textsl{TL}_4$ would require $q<0$, and thus do not exist, according to our convention $q>0$.

\emph{Perfect morphisms for trivalent categories.---}Our construction of perfect tangles for the Temperley-Lieb algebra automatically provides us with a ready supply of examples for any modular tensor category containing it as a subalgebra. These morphisms, by definition, do not include trivalent vertices. 

With the appropriate modification of the definition for \emph{perfectness} in the case of an uneven number of legs, there are then \emph{perfect morphisms} that nontrivially involve the trivalent vertex. The most obvious is the trivalent vertex itself. We focus here on perfect morphisms in $\mathcal{C}_4 \equiv \mathrm{Hom}\left( 1, X^{\otimes 4} \right)$ because, as we'll see, we'll be able to inductively construct other perfect morphisms in $\mathcal{C}_{2k}$ for $k\geq 2$. No obvious construction for odd $n$ is currently known.
 
Calling the loop and triangle values $d$ and $t$, respectively, we restrict our attention to \emph{cubic categories}, i.e.\ trivalent categories with $\dim \mathcal{C}_4=4$ and such that the diagrams without internal faces give a basis for that hom-space \cite{morrison_categories_2017}. There is thus a relation for the square, to wit
\begin{align*}
\begin{tikzpicture}[scale = 0.5, baseline=-1mm]
	\draw (-1,1)  -- (-0.7,0.7) -- (0.7,0.7) -- (1,1);
	\draw (-1,-1)  -- (-0.7,-0.7) -- (0.7,-0.7) -- (1,-1);	
	\draw (-0.7,0.7) -- (-0.7,-0.7);
	\draw (0.7,0.7) -- (0.7,-0.7);	
\end{tikzpicture}
=
A\cdot \left( \,
	\begin{tikzpicture}[scale=0.5, baseline=-1mm]
		\begin{scope}
			\foreach \x in {0, 1} \draw (\x, -1) -- (\x, 1);
		\end{scope}
		\begin{scope}[shift={(2.2,0)}]
			\node (beta) at (-0.5,0) {$+$};
			\draw (0,1) arc[start angle=180, end angle=360, radius=5mm];
			\draw (0,-1) arc[start angle=180, end angle=0, radius=5mm];
		\end{scope}
	\end{tikzpicture}
\, \right)
+ B\cdot \left( \,
	\begin{tikzpicture}[scale=0.5, baseline=-1mm]
		\begin{scope}[shift={(0,0)}]
			\draw (-0.8,1) -- (-0.4, 0) -- (0.4,0) -- (0.8,1);
			\draw (-0.8,-1) -- (-0.4, 0) -- (0.4,0) -- (0.8,-1);
		\end{scope}
		\begin{scope}[shift={(2.2,0)}]
			\node (delta) at (-1,0) {$+$};
			\draw (0.7, 1) -- (0, 0.5) -- (-0.7, 1);
			\draw (0.7, -1) -- (0, -0.5) -- (-0.7, -1);
			\draw (0, 0.5) -- (0, -0.5);
		\end{scope}
	\end{tikzpicture}
	\, \right),
\end{align*}
where $A = \frac{d t^2 + t^2 -1}{dt + d +t}$ and $B = \frac{-t^2 + t+1}{dt + d +t}$. 
Consequently, it is no longer true that a morphism is perfect only if it has support on the identity.

A general element in $\mathcal{C}_4$ is of the form
\begin{align*}
T \equiv  \,
\begin{tikzpicture}[scale=0.5, baseline=-1mm]
	\begin{scope}
		\node (alpha) at (-0.5,0) {$\alpha$};
		\foreach \x in {0, 1} \draw (\x, -1) -- (\x, 1);
	\end{scope}
	\begin{scope}[shift={(2.2,0)}]
		\node (beta) at (-0.5,0) {$+\,\beta$};
		\draw (0,1) arc[start angle=180, end angle=360, radius=5mm];
		\draw (0,-1) arc[start angle=180, end angle=0, radius=5mm];
	\end{scope}\, 
	\begin{scope}[shift={(5.2,0)}]
		\node (gamma) at (-1.4,0) {$+\,\gamma$};
		\draw (-0.8,1) -- (-0.4, 0) -- (0.4,0) -- (0.8,1);
		\draw (-0.8,-1) -- (-0.4, 0) -- (0.4,0) -- (0.8,-1);
	\end{scope}
	\begin{scope}[shift={(7.5,0)}]
		\node (delta) at (-1,0) {$+\,\delta$};
		\draw (0.7, 1) -- (0, 0.5) -- (-0.7, 1);
		\draw (0.7, -1) -- (0, -0.5) -- (-0.7, -1);
		\draw (0, 0.5) -- (0, -0.5);
	\end{scope}
\end{tikzpicture}
,
\end{align*}
and we descibe solutions only in terms of coefficients. The morphism $T$ is invariant under rotation by $\pi$ and taking the adjoint results in conjugation of the coefficients. One rotation exchanges $\alpha$ and $\beta$, and $\gamma$ and $\delta$. There are hence $8$ real equations that must be satisfied by a perfect morphism \footnote{The 8 equations are 
$\lvert \alpha \rvert^2 + A \lvert \gamma \rvert^2 \neq 0$, $\lvert \beta \rvert^2 + A \lvert \delta \rvert^2 \neq 0$,
$d\lvert \beta \rvert^2 + (\alpha+\gamma, \beta)+ A \lvert \gamma \rvert^2 = 0$, $d\lvert \alpha \rvert^2 + (\beta+\delta, \alpha)+ A \lvert \delta \rvert^2 = 0$, $(\alpha,\gamma) + B\lvert \gamma \rvert^2 = 0$, $(\beta,\delta) + B\lvert \delta \rvert^2 = 0$,
$\lvert \delta \rvert^2 + (\alpha+t \gamma, \delta) + B\lvert \gamma \rvert^2 = 0$, $\lvert \gamma \rvert^2 + (\beta+t \delta, \gamma) + B\lvert \delta \rvert^2 = 0$, where we've employed the notation $(\alpha, \beta) \equiv \alpha\overline{\beta} + \overline{\alpha}\beta = 2 \left( \Re{\alpha}\Re{\beta} + \Im{\alpha}\Im{\beta}\right) \in\mathbb{R}$. These equations are real, and, since $(\alpha,\beta) = (\overline{\alpha},\overline{\beta})$, indeed cover all rotations.}
.

The Fibonacci category obeys additional relations, so that $\dim\mathcal{C}_4 = 2$, i.e., a general element in $\mathcal{C}_4$ is an element of $\textsl{TL}_2$. As a consequence we have already described all the available solutions for this case. The first nontrivial example is given by \emph{Kuperberg's Spider} $G_2$, in particular, the cubic category $\left( G_2 \right)_q$ at $q$ a primitive $8$th root of unity \cite{morrison_categories_2017}. In such $\left( G_2 \right)_q$-categories,
\begin{align*}
d = q^{10}+ q^8+ q^2+ 1 + q^{-2}+ q^{-8}+ q^{-10}
\end{align*}
and
\begin{align*}
t = - \frac{q^2 - 1 + q^{-2}}{q^4 + q^{-4}}.
\end{align*}
(Due to restricted computational power, a general solution for arbitrary roots of unity $q$ could not be produced so far.)

If $q$ is an $8$th root of unity, then one computes 
\begin{align*}
d = 3,\quad t= -\frac{1}{2}, \quad A = \frac{1}{4}, \quad B = 0,
\end{align*}
and we can give solutions for perfect morphisms. If $\alpha \neq 0$ we may normalize to $\alpha = 1$. If, on the other hand, $\alpha = 0$, one requires that $\delta = 0$ and, without loss of generality, we normalize so that $\beta = 1$. The two simplest solutions are then given by 
\begin{align*}
&
\alpha = 1, \qquad 
\beta = 0, \qquad
\gamma = 0, \qquad
\delta = -2 \\
\text{or}\quad&
\alpha = 0, \qquad 
\beta = 1, \qquad
\gamma = -2, \qquad
\delta = 0,
\end{align*}
which are rotations of one another. The first one is actually just the $r=0$ case of the simple solution
\begin{align*}
\alpha = 1, \qquad \beta = i r, \qquad \gamma = -2 i r ,\qquad \delta = -2,
\end{align*}
valid for $r\in\mathbb{R}$, and the second is the only solution with $\alpha = 0$, $\beta=1$. Finally, for $r \in (-4,0)$, we get a last set of solutions
\begin{multline*}
\alpha = 1, \quad 
\beta = r \pm i (r - 1) f(r) ,\quad
\gamma = \pm i 10  f(r) , \\ \text{and} \quad
\delta = \mp i \frac{2\beta}{f(r)},
\end{multline*}
where $f(r)\equiv \sqrt{\frac{-r}{4+r}}$.

The final example we study is the \emph{Haagerup fusion category} $H3$ \cite{morrison_categories_2017} which has, by convention, loop and triangle values
\begin{align*}
d = \frac{3 +\sqrt{13}}{2},\qquad t = \frac{2-\sqrt{13}}{3}.
\end{align*}
It is cubic, so we employ the equations we had before to find solutions. We found that a perfect morphism needs to be supported on the identity. Two solutions are then given by
\begin{multline*}
\alpha = \frac{4}{ \sqrt{3}}, \quad \beta =  \frac{d}{\sqrt{3}} \mp i \sqrt{5 - d}, \\ \gamma = - 2 \sqrt{3}  (d + 1) \pm i2 \sqrt{5d + 2},\\ \text{and}\quad \delta = - 3 \sqrt{3} d \pm i \sqrt{5 - d},
\end{multline*}
where we did not normalize to $\alpha = 1$ for readability reasons. 

\emph{An inductive construction.---}Here we describe an inductive procedure proposed by Vaughan Jones whereby one can produce larger perfect tangles from smaller perfect tangles. This construction is reminiscent of that explored by Reutter and Vicary \cite{reutter_biunitary_2016} to build biunitary connections. The intuition for the construction is entirely captured by the following \emph{braid heuristic}: Consider the picture
\begin{equation}
	A \,\sim \raisebox{-1cm}{\includegraphics{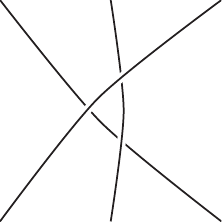}},
\end{equation}
which represents a tangle in some (braided) category, with the over- and under-crossings given by braidings. One can confirm that this tangle is indeed perfect. The idea behind the inductive construction is to replace the braidings by perfect $2$-tangles. Surprisingly this works, even for categories without braiding. To demonstrate this assertion we fix some perfect $2$-tangle $T$ and inductively define $\widetilde{T}_n$, $n\ge 3$ according to
\begin{equation*}
	\begin{tikzpicture}[baseline=-0.5mm]
		\node (T) at (-2,0) {$\widetilde{T}_n \equiv\,$};
		\node[draw] (T1) at (0,0) {$\widetilde{T}_{n-1}$};
		\node[draw] (T2) at (2,0) {$T$};
		\path ($(T1.north) + (-0.2,0)$) edge node[left]{$n-2$} ($(T1.north) + (-0.2,0.5)$)
			($(T1.south)$) edge node[left]{$n-1$} ($(T1.south) + (0,-0.5)$)
			($(T2.north) + (-0.2,0)$) edge ($(T2.north) + (-0.2,0.5)$)
			($(T2.north) + (0.2,0)$) edge ($(T2.north) + (0.2,0.5)$)
			($(T2.south) + (0.2,0)$) edge ($(T2.south) + (0.2,-0.5)$);		
		\draw  ($(T1.north) + (0.2,0)$) .. controls (1.2,1.5) and (0.8,-1.5).. ($(T2.south) + (-0.2,0)$);
	\end{tikzpicture}\quad .
\end{equation*}
Using $\widetilde{T}_n$ we now construct the tangle $A_n$ via
\begin{align*}
A_n \equiv \quad
	\begin{tikzpicture}[baseline = -1mm]
		\node[draw] (T) at (0,0.7) {$\widetilde{T}_{n}$};
		\node[draw] (B) at (0.5, -.7) {$A_{n-1}$};
		\path (T.north) edge ($(T.north) + (0,0.5)$) 
			($(T.south) + (-0.2, 0)$) edge ($(T.south |- B.south) + (-0.2, -0.5)$)
			($(T.south) + (0.2,0)$) edge node[right]{$n-1$} (B.north)
			($(B.south)$) edge ($(B.south |- B.south) + (-0.0, -0.5)$);
	\end{tikzpicture}\quad .
\end{align*}
One can confirm that $A_n$ so defined yields a perfect tangle (the proof of this may be found in the supplementary material).

\emph{A conjecture.---}
Emboldened by the success of the braid heuristic, it is plausible that the following \emph{horizontal construction}, which is illustrated in the next picture, will always yield a perfect tangle. Fix an $n$-tangle $T$ and an $m$-tangle $S$, both perfect, and a perfect $2$-tangle $R$. Construct an $(n+m)$-tangle by first building $T\otimes S$, i.e.\ draw the tangles next to each other, as boxes in standard form. Then braid the bottom legs in any way you want, but subject to the requirement that strings coming from the same box don't cross. Next, mirror the same braiding for the upper legs as seen in the graphic. Finally, replace \emph{all} crossings by $R$, disregarding orientation. If we call the braiding of the bottom legs $\mathcal{L}$, then the tangle obtained via the above procedure shall be denoted by the symbol $T \oslash_\mathcal{L} S$.
\begin{equation*}
T \oslash_\mathcal{L} S \quad\equiv\quad
\begin{tikzpicture}[baseline=-0.1cm, scale=2]
	\node[draw, blue,line width=0.5pt] (A) at (-1,0) {$\phantom{\hat{ABC}^2}$};
	\node[draw, red,line width=0.5pt] (B) at (1,0) {$\phantom{\hat{ABC}^2}$};

	\draw[blue,line width=0.5pt] ($(A.north) - (0.2,0)$) .. controls +(up:1mm)  .. + (0.5, 1);
	\draw[blue,line width=0.5pt] ($(A.north) - (0,0)$) .. controls +(up:1mm)  .. + (1.2, 1);
	\draw[blue,line width=0.5pt] ($(A.north) + (0.2,0)$) .. controls +(up:1mm)  .. + (2, 1);
	
	\draw[blue,line width=0.5pt] ($(A.south) - (0.2,0)$) -- +(0, -1);
	\draw[blue,line width=0.5pt] ($(A.south) - (0,0)$) .. controls +(down:1mm)  .. +(0.8, -1);
	\draw[blue,line width=0.5pt] ($(A.south) + (0.2,0)$) .. controls +(down:1mm)  .. +(1.8, -1);
	
	\draw[line width=5pt, white] ($(B.north) + (0.2,0)$) .. controls +(up:1mm)  ..  + (-0.5, 1);
	\draw[line width=5pt, white] ($(B.north) - (0,0)$) .. controls +(up:1mm)  .. + (-1.2, 1);
	\draw[line width=5pt, white] ($(B.north) - (0.2,0)$) .. controls +(up:1mm)  .. + (-2, 1);	
	\draw[red,line width=0.5pt] ($(B.north) + (0.2,0)$) .. controls +(up:1mm)  .. + (-0.5, 1);
	\draw[red,line width=0.5pt] ($(B.north) - (0,0)$) .. controls +(up:1mm)  .. + (-1.2, 1);
	\draw[red,line width=0.5pt] ($(B.north) - (0.2,0)$) .. controls +(up:1mm)  .. + (-2, 1);
	
	\draw[line width=5pt, white] ($(B.south) + (0.2,0)$) -- +(0, -1);
	\draw[line width=5pt, white] ($(B.south) - (0,0)$) .. controls +(down:1mm)  .. +(-0.8, -1);
	\draw[line width=5pt, white] ($(B.south) - (0.2,0)$) .. controls +(down:1mm)  .. +(-1.8, -1);
	\draw[red,line width=0.5pt] ($(B.south) + (0.2,0)$) -- +(0, -1);
	\draw[red,line width=0.5pt] ($(B.south) - (0,0)$) .. controls +(down:1mm)  .. +(-0.8, -1);
	\draw[red,line width=0.5pt] ($(B.south) - (0.2,0)$) .. controls +(down:1mm)  .. +(-1.8, -1);
\end{tikzpicture}\ .
\end{equation*}
One can verify this indeed yields a perfect tangle in this case. We don't have as yet a proof that the general procedure always yields a perfect tangle, however, we have no counterexamples. 

\emph{Conclusions and outlook.---}We have generalised the notion of a perfect tensor to categories other than quantum spin systems, particularly, certain modular tensor categories which are understood to model the kinematics of systems of anyons. Such perfect tangles and perfect morphisms are expected to give rise to microscopic models of the AdS/CFT correspondence whose quasiparticles exhibit anyonic statistics, i.e., generalisations of the Chern-Simons/Wess-Zumino-Witten correspondence. There are many open questions, including, how to prove that the horizontal construction always works, generalising our constructions to odd numbers of legs, and building perfect morphisms for MTCs other than cubic categories.

\acknowledgments{Inspiring conversations with Vaughan Jones, who contributed the braid heuristic argument, are gratefully acknowledged. Thanks also to Ramona Wolf and Deniz Stiegemann for helpful comments and suggestions. This work was supported by the the DFG through SFB 1227 (DQ-mat) and the RTG 1991.}

\bibliography{perfecttangles}

\widetext
\appendix
\section{Supplementary material}
\subsection{Category theory}
In order for this Letter be as self-contained as possible -- and to normalise notations and conventions -- we include here a short introduction to category theory.

In this work we only ever need small categories, hence our usage of \emph{sets} rather than \emph{classes} in the following definition.
\begin{definition}[Category]\index{Category} A \emph{category} $\mathbf{C}$ consists of a set $\mathsf{obj}(\mathbf{C})$ of \emph{objects}, and a set $\mathsf{mor}(\mathbf{C})$ of \emph{morphisms}, subject to the following axioms. To each morphism we assign two objects called the \emph{source} and the \emph{target}, respectively. If $f\in\mathsf{mor}(\catname{C})$ with source $A$ and target $B$, this is denoted $A\stackrel{f}{\to} B$ or $f:A\rightarrow B$.

The set of morphisms is required to be closed under \emph{composition}, denoted $\circ$: If $A\morphism{f}B$ and $B\morphism{g}C$, then there exists a unique morphism $A\morphism{m}C$ s.t.\ $m = g\circ f$. In the language of diagrams used in category theory this uniqueness assertion is commonly denoted by a dashed or dotted arrow in a commutative diagram:

\begin{center}
\begin{tikzcd}
A \ar[r,"f"] \ar[dashrightarrow, swap]{dr}{m}& B \ar[d,  "g"]\\
  & C
\end{tikzcd}
\end{center}
Composition of morphisms must be associative, i.e.\ $h \circ \left( g \circ f\right) = \left( h\circ g \right)\circ f$. This associativity requirement allows us to drop the parentheses and we may just write $h\circ g\circ f$. 

For each object $A$ there is a unique morphism $\id_A$, the identity, satisfying $\id_A\circ f = f$, $g\circ \id_A = g$, for each $f$ with target $A$ and each $g$ with source $A$. This morphism will occasionally also be written $1_A$.
\end{definition}

The set of all morphisms from an object $A$ to an object $B$ is denoted by $\mathbf{C}[A,B]$, or $\mathrm{Hom}_\mathbf{C}\left( A,B \right)$, where $\mathrm{Hom}$ stands for \emph{hom-set}. Various other notations are also used in the literature, in particular, we often refer to $\mathrm{Hom}_\mathbf{C}\left( A,B \right)$ as a \emph{morphism space}.

The concept of a category is easily understood by looking at a few examples. The easiest that comes to mind is the category $\mathbf{Set}$, where the objects are just regular old sets, and the morphisms are functions. The composition of morphisms is thus simply the composition of functions, and the identity morphism on each set is the identity function.  It is then elementary to check the category axioms.

Another oft-cited example is this: a group can be seen as a category with one object -- say, $\star$ --  where every morphism is invertible. The group operation is composition of morphisms, and \emph{invertibility} here means that for every morphism $g$ there exists a unique morphism $g\inv$ that is the two-sided inverse of $g$, i.e.\ $g\circ g\inv = \id_\star = g\inv\circ g$. This is also precisely the definition of \emph{isomorphism}: A morphism which has a two-sided inverse.

Two special kinds of morphism need to be mentioned.  Firstly, $f$ is called \emph{an epimorphism} (or \emph{epic)} iff
\begin{align*}
g_1 \circ f = g_2 \circ f \implies g_1 = g_2,
\end{align*}
and $f$ is called \emph{a monomorphism} (or \emph{monic}), often denoted $X\hookrightarrow Y$, iff
\begin{align*}
f\circ g_1 = f\circ g_2 \implies g_1 = g_2.
\end{align*}
These in a sense generalize the notions of \emph{surjective} and \emph{injective}, respectively. An isomorphism is in particular always both monic and epic.

In the two examples above morphisms were still functions. But this need not be the case. Suppose, for example, that $S$ is a poset and let $x,y,z\in S$. We interpret $S$ as a category and define hom-sets by $S[x,y]\equiv\{(x,y) \}$ if $x \leq y$, $\emptyset$ otherwise. Then every hom-set has at most one element and composition $(y,z)\circ(x,y) = (x,z)$ is correctly defined.

From categories $\catname{C}, \catname{D}$ multiple new categories can be built. The \emph{opposite category} $\catname{C}^{op}$, for example, has the same objects but for each morphism source and target are exchanged. Or one may define a \emph{product category} $\catname{C}\times\catname{D}$, where the objects are the elements of $\mathsf{obj}(\catname{C})\times\mathsf{obj}(\catname{D})$, and morphisms are the elements of $\mathsf{mor}(\catname{C})\times\mathsf{mor}(\catname{D})$. Composition is inherited from the original categories, i.e.\ $(f_2,g_2)\circ(f_1,g_1) = (f_2\circ f_1, g_2\circ g_1)$.

A special property is \emph{linearity}: a category $\catname{C}$ is called \emph{$\mathbb{F}$-linear} if for all objects $X,Y$ the hom-set $\catname{C}[X,Y]$ is a (finite-dimensional) vector space over the field $\mathbb{F}$.

Some categories possess \index{Object!initial}\index{Object!terminal}\emph{initial} (\emph{terminal}) objects. An initial (terminal) object $A$ in $\mathbf{C}$ is one where for all objects $B$ the hom-set $\mathbf{C}[A,B]$ ($\mathbf{C}[B,A]$) consists of only one element. An object that is both initial and terminal is referred to as a \emph{zero object}. All of these are defined up to isomorphism. As an example, $\catname{Set}$ possesses no zero object, but the empty set serves as initial object, and every singleton is terminal.

If a category does have a zero object $0$, then an object $X$ is called \emph{simple} if the only objects that have monomorphisms into $X$ are $0$ and $X$ itself.

An elementary concept is that of \emph{functors}, which may be viewed as morphisms between categories.

\begin{definition}[Functor]{\index{Functor}}
Let $\catname{C},\catname{D}$ be categories. A \emph{functor} $F:\catname{C}\rightarrow \catname{D}$ assigns to each object (morphism) in $\catname{C}$ an object (morphism) in $\catname{D}$ s.t.\ $F\id_A = \id_{FA}$, and either
\begin{itemize}
\item[•]$FA\morphism{Ff} FB$ or
\item[•]$FB\morphism{Ff} FA$,
\end{itemize}
 for $f\in\catname{C}[A,B]$.
In the former case $F$ is called \emph{covariant}, in the latter \emph{contravariant}. A covariant functor respects composition,
\begin{align*}
F(g\circ f) = Fg\circ Ff,
\end{align*}
while a contravariant one reverses it,
\begin{align*}
F(g\circ f) = Ff\circ Fg.
\end{align*}
\end{definition}
Functors are ubiquitous. One extremely important example is the so-called \emph{hom-functor}, which exists in both flavors, co- and contravariant.

\begin{example}[Hom-functor]
Let $\catname{C}$ be any small category, and fix an object $X$. The covariant hom-functor $\catname{C}[X,-]:\catname{C}\rightarrow \catname{Set}$ sends each object $A$ to the hom-set $\catname{C}[X,A]$, and each morphism $f\in\catname{C}[A,B]$ gets sent to the function
\begin{align*}
\catname{C}[X,f]:\catname{C}[X,A] &\rightarrow \catname{C}[X,B],\\
l &\mapsto f\circ l
\end{align*}
The associativity of function composition then shows the covariance.

One analogously defines the contravariant hom-functor $\catname{C}[-,X]$.
\end{example}
Of course there are also easier examples. The trivial functor from any category to the category with one object and one morphism, operating in an obvious way; or the identity functor $\mathbf{1}$, an endofunctor (a functor from a category to itself), sending each object and each morphism to itself.

Observing the existence of \emph{natural transformations} was crucial in leading Eilenberg and MacLane \cite{eilenberg_general_1945} to starting this branch of mathematics. These entities at the very heart of category theory can be regarded as morphisms between functors. The following definition makes this precise.

\begin{definition}[Natural Transformation]
Let $F,G:\catname{C}\rightarrow\catname{D}$ be functors. A \emph{natural transformation} $\eta$ from $F$ to $G$, written $\eta:F\Rightarrow G$, is a family of morphisms $(\eta_X)_{X\in \catname{C}}$, $FX\morphism{\eta_X} GX$, indexed by objects of $\catname{C}$. The morphism $\eta_X$ is called the \emph{component at $X$}, and the components are required to make the following diagram commute, for all $f\in \catname{C}[X,Y]$.

\begin{figure}[!htp]\centering
\begin{tikzcd}
FX \ar[r,"Ff"] \ar[d,swap,"\eta_X"] & FY \ar[d, "\eta_Y"]\\
GX \ar[r, "Gf"] & GY
\end{tikzcd}
\end{figure}
In words: $Gf\circ\eta_X = \eta_Y \circ Ff$. In this way a natural transformation relates two functors. If all $\eta_X$ are isomorphisms, $\eta$ is called a \emph{natural isomorphism}.
\end{definition}
Composition of natural transformation is defined component-wise, and the collection of all functors from $\catname{C}$ to $\catname{D}$, with natural transformations as morphisms, then forms the \emph{functor category} $\catname{D}^\catname{C}$ or $[\catname{C},\catname{D}]$.

Now that we have all the basic gadgetry at our disposal we can introduce a categorical version of the \emph{tensor product}: a \emph{monoidal structure}. All categories in this Letter possess one of those.

\begin{definition}[Monoidal Category]
A category $\catname{C}$ equipped with a \emph{monoidal structure} is called a \emph{monoidal category}. The monoidal structure consists of the \emph{monoidal unit} $\mathbf{I}$, a functor $\otimes:\catname{C}\times\catname{C}\rightarrow\catname{C}$, usually written as an infix, and three natural isomorphisms satisfying certain so-called \emph{coherence conditions}. The isomorphisms are
\begin{itemize}
\item The \emph{associator} $\alpha: \otimes \circ \left(\otimes\times \mathbf{1}\right) \Rightarrow \otimes \circ \left(\mathbf{1}\times \otimes\right)$ with components $\alpha_{A,B,C}:(A\otimes B)\otimes C\rightarrow A\otimes\left( B\otimes C \right)$.
\item The \emph{right unitor} $\rho: -\otimes\mathbf{I}\Rightarrow \mathbf{1}$ with components $\rho_A:A\otimes \mathbf{I}\rightarrow A$.
\item The \emph{left unitor} $\lambda: \mathbf{I}\otimes - \Rightarrow \mathbf{1}$ with components $\lambda_A: \mathbf{I}\otimes A\rightarrow A$.
\end{itemize}
The coherence conditions are usually given in the form of diagrams. The \emph{associativity coherence}, also known as the \emph{pentagon axiom}, is:
\begin{center}
\tikzsetnextfilename{pentagon_axiom}
\begin{tikzpicture}[commutative diagrams/every diagram]
\node (P0) at (90:2.8cm) {$\left( \left( A\otimes B \right) \otimes C \right) \otimes D $};
\node (P1) at (90+1*72:2.5cm) {$\left( A\otimes \left( B\otimes C \right) \right)\otimes D$};
\node (P2) at (90+2*72:2.5cm) {\makebox[7ex][r]{$A\otimes\left(\left( B\otimes C \right) \otimes D \right)$}};
\node (P3) at (90+3*72:2.5cm) {\makebox[5ex][l]{$A\otimes\left( B\otimes \left( C\otimes D \right) \right)$}};
\node (P4) at (90+4*72:2.5cm) {$\left( A\otimes B \right)\otimes\left( C\otimes D \right)$};

\path[commutative diagrams/.cd, every arrow, every label]
	(P0) edge node[swap] {$\alpha_{A, B,C}\otimes 1_D$} (P1)
	(P1) edge node[swap] {$\alpha_{A,B\otimes C, D}$} (P2)
	(P2) edge node {$1_A\otimes \alpha_{B,C,D}$} (P3)
	(P0) edge node {$\alpha_{A\otimes B, C, D}$} (P4)
	(P4) edge node {$\alpha_{A,B,C\otimes D}$} (P3);
\end{tikzpicture}
\end{center}
The \emph{triangle axiom} is
\begin{center}
	\begin{tikzcd}
		\left( A\otimes \mathbf{I} \right)\otimes B \ar[swap, ddr, "\rho_A\otimes 1_B"] \ar[r, "\alpha_{A,\mathbf{I},B}"] & A\otimes\left( \mathbf{I}\otimes B \right)
		\ar[dd, "1_A\otimes \lambda_B"]		
		\\ & \\
		 & A\otimes B
	\end{tikzcd}
\end{center}
\end{definition}

The two main axioms in this definition give rise to MacLane's \emph{coherence theorem} for monoidal categories \cite{mac_lane_categories_1998}, which tells us that in a monoidal category every diagram consisting of only associators and unitors commutes. In turn we may drop the parentheses in monoidal products of objects, since all ways of bracketing are isomorphic \cite{etingof_tensor_2015}. 

We will also only be interested in strict monoidal categories, which means that the associator and both unitors  are taken to be the identity natural transformation, i.e.\ all components are identity morphisms.

Monoidal categories are also called \emph{tensor categories}, a name coming from the prototypical example: $\catname{Vect}_\mathbb{F}$, the category of vector spaces over $\mathbb{F}$, where the monoidal product is just your usual tensor product, and the field itself serves as the monoidal unit. The morphisms in this category are linear maps between vector spaces, so $\catname{Vect}_\mathbb{F}$ is also an example of a linear category.

 Given a monoidal category, we may define \emph{dual objects}, whence a whole variety of new terminology emerges.

\begin{definition}\index{Object!dual}\index{Category!rigid}\index{Category!pivotal}
\begin{itemize}
\item[{(i)}] A monoidal category is called \emph{symmetric} if for all objects $X,Y$ there exists an isomorphism $s_{XY}:X\otimes Y \rightarrow Y\otimes X$, natural in $X,Y$, such that
\[s_{YX}\inv = s_{XY}^{\phantom{-1}}\]
and such that an obvious associativity coherence diagram between $(X\otimes Y)\otimes Z$ and $Y\otimes (Z\otimes X)$ commutes \footnote{The parentheses were retained here to clarify how the associator would come into play if the category were not strict.}.
\item[{(ii)}] Let $\catname{C}$ be a monoidal category. Let $X\in\mathsf{obj}(\catname{C})$. If there exists an object $X^*$ together with two morphisms
\begin{align*}
\eta:\mathbf{1} \rightarrow X\otimes X^*
\end{align*}
and 
\begin{align*}
\epsilon:X^*\otimes X \rightarrow \mathbf{1},
\end{align*}
such that both diagrams in Figure \ref{fig:ev_coev_triangle} commute, then we say that $X$ has a \emph{left dual}, namely $X^*$. $\epsilon$ is called the \emph{evaluation}, while $\eta$ is the \emph{coevaluation}.
\begin{figure}[!htp]\centering
\begin{tikzcd}
X \ar[dr,swap,"1_X"]  \ar[r, "\eta\otimes 1_X"] & X\otimes X^*\otimes X  \ar[d,"1_X\otimes \epsilon"]\\
& X
\end{tikzcd}~~
and~~
\begin{tikzcd}
X^* \ar[dr,swap,"1_{X^*}"]  \ar[r, "1_{X^*}\otimes \eta"] &X^*\otimes X\otimes X^*  \ar[d,"\epsilon\otimes1_X"]\\
& X^*
\end{tikzcd}
\caption{The axioms for the evaluation and the coevalution.}\label{fig:ev_coev_triangle}
\end{figure}

Right duals ${}^*X$ are defined similarly (note that there are different conventions in use regarding which dual is \emph{right} and which is \emph{left}). Using string diagrams (read from bottom to top) the evaluation and the coevaluation can be drawn as caps and cups, respectively. An object is \emph{self-dual} if $X\cong X^*$.

\item[{(iii)}] A monoidal category is called \emph{rigid} if every object has duals on both sides. Taking left (right) duals then becomes a contravariant endofunctor, with \emph{dualization} of morphism given by
\begin{align*}
&\left( X\morphism{f}Y \right)^*
=\, Y^*\morphism{f^*}X^*
\\
&~~\equiv~  Y^* \morphism{ 1_{Y^*} \otimes \eta_X} Y^* \otimes X\otimes X^*  \morphism{ 1_{Y^*}\otimes f\otimes 1_{X^*}} Y^*\otimes Y \otimes X^* \morphism{\epsilon_Y \otimes 1_{X^*}} X^*.
\end{align*}
In the language of string diagrams this is just rotating the morphism by $180^\circ$ \cite{wang_topological_2010}.
\item[{(iv)}] A rigid category is called \emph{pivotal} if for each object $X$ we get an isomorphism $\phi_X:X \rightarrow (X^*)^*$ natural in $X$ and commuting with $\otimes$ in the obvious sense. The morphisms $\phi$ are called the \emph{pivotal structure}.
\end{itemize}
\end{definition}

If an object $X$ has two duals $X^*,$ $X^\vee$ then it is easy to construct an isomorphism $X^*\cong X^\vee$ using only the evalutaions and coevaluations. Categories with duals admit a nice graphical calculus, called \emph{string diagrams}, rigorously proved by \cite{joyal_geometry_1991} first, which we are only going to graze very lightly. For a full survey, see either the original paper (which is extremely technical), or the excellent paper \cite{selinger_survey_2010}.

In essence, a string diagram depicts a morphism $f:X\rightarrow Y$ as a box labelled $f$ with two strings attached, one labelled $X$ and one labelled $Y$, drawn `from bottom to top'.
Tensor product is juxtaposition of both strings and morphisms.  The monoidal unit is not depicted, and the identity morphism on $X$ is just the string labelled $X$. Finally, if $X$ has a (left) dual, then the evaluation and coevaluation maps are drawn as caps and cups, respectively. That is
\begin{align*}
\epsilon_X \doteq \,
\begin{tikzpicture}[scale=0.7,baseline=1mm]
	\draw (0,0)node[left]{$X^*$} arc[start angle=180, end angle=0, radius=1cm] node[right]{$X$};
\end{tikzpicture}\, , \qquad
\eta_X \doteq \,
\begin{tikzpicture}[scale=0.7,baseline=-2mm]
	\draw (0,0)node[left]{$X$} arc[start angle=180, end angle=360, radius=1cm] node[right]{$X^*$};
\end{tikzpicture}\,.
\end{align*}
This makes clear why we also speak of \emph{taking duals} as \emph{rotation by $180^\circ$}.

In a category with pivotal structure $\phi$, which we from now on assume to be trivial, we can then define a notion of \emph{left} and \emph{right trace} for a morphism $f:X\rightarrow X$ by setting
\tikzexternaldisable
\begin{align*}
\mathrm{tr}_l(f)&\equiv\,
\begin{tikzpicture}[scale=0.5, baseline=-1mm]
	\node[draw] (f)  at (0,0) {$f$};
	\draw (f.north) arc[start angle=0, end angle=180, radius=5mm] -- 
		 ($(f.south) + (-1,0)$) arc[start angle=180, end angle=360, radius=5mm];
\end{tikzpicture}\\
&= 1 
\morphism{\eta_{X^*}} X^*\otimes X^{**} = X^*\otimes X
\morphism{1_{X^*}\otimes f} X^*\otimes X
\morphism{\epsilon_X} 1
\end{align*}
and
\begin{align*}
\mathrm{tr}_r(f)&\equiv\,
\begin{tikzpicture}[scale=0.5, baseline=-1mm]
	\node[draw] (f)  at (0,0) {$f$};
	\draw (f.north) arc[start angle=180, end angle=0, radius=5mm] -- 
		 ($(f.south) + (1,0)$) arc[start angle=0, end angle=-180, radius=5mm];
\end{tikzpicture}\,.
\end{align*}

In general there is no reason for the left and the right trace  to coincide, but if $\mathrm{tr}_r = \mathrm{tr}_l$, then the category is called \emph{spherical}, the name coming from the fact that one trace diagram can be transformed into the other by diffeomorphism on the 2-sphere. All categories we encounter here are, in fact, spherical.

\subsection{Planar algebras}
Although we mostly work with Temperley-Lieb (TL) algebras, we introduce here the notion of a \emph{planar algebra} (PA) \cite{jones_planar_1999}, introduced by Jones as a sweeping generalisation of subfactor theory. One of the reasons for this is that the collection of TL algebras can be equipped with a planar structure, yielding the Temperley-Lieb planar algebra, and that the class of trivalent categories have a deep connection to planar algebras as well. The other reason is that our main definitions and some theorems work in planar algebras in general, that is, independent of the nature of the things we ``put into the boxes''.

Planar algebras are designed to graphically capture multilinearity, not unlike the diagrammatic language of (rigid) monoidal categories \cite{joyal_geometry_1991} or other pictorial computation aids, such as found in the theory of tensor networks, or their antecedant Penrose diagrams. We consider examples almost exclusively in Temperley-Lieb algebras. This is because the TL \emph{planar} algebra is initial in the category of subfactor planar algebras \cite{peters_planar_2010}.

\subsubsection{Planar tangles}
Here we introduce the basic notion of a \emph{planar tangle}, and give a few examples. We then describe the operadic nature of these objects. The planar tangle definition given here is somewhat intuitionistic, and the reader desiring more rigour is invited to follow the original references \cite{jones_planar_1999}.

Our first description of \emph{shaded} planar tangles is in terms of disks, which eventually become rectangles, only to once again become circles when we briefly discuss \emph{unshaded} planar tangles. We advise the reader to look at \figref{fig:Tangle Examples1}, and to maybe draw their own pictures.

\begin{definition}[Shaded Planar Tangle]\index{Planar Tangle!shaded}\label{def:shaded planar tangle}
Define the set \[\mathfrak{Col} \equiv \mathbb{N}_0\times \{+,-\}\] of \emph{colors}, and pick an element $(n_0,\epsilon_0)\in\mathfrak{Col}$. A \emph{shaded $(n_0, \epsilon_0)$-tangle} $T$ or \emph{shaded tangle} $T$ \emph{with color $(n_0, \epsilon_0)$} is then a disk $D_0$ together with the following extra data.
\begin{itemize}
\item[{(1)}] In the interior there is a (possibly empty) collection $\mathcal{D}_T$ of disjoint disks $D_i$, $1\leq i \leq \lvert \mathcal{D}_T \rvert$, each with \emph{color} $(n_i, \epsilon_i)$.
\item[{(2)}] There are $2n_i$ distinct points on the boundary of each disk. Also, for each disk, one of these points is marked with a $\markedPoint$ and called the \emph{first} or \emph{marked point}.
\item[{(3)}] All distinct points are the endpoints of smooth non-intersecting curves in the interior of $D_0$, and there may also be closed curves not attached to the boundary of any disk. The strings have to be \emph{compatible} with the marked points of all disks in the following sense. Label the marked point of the disk $D_i$ with $\epsilon_i$, then label all other distinct points in a clockwise fashion alternating between $+$ and $-$. A string with one endpoint on an \emph{internal} disk and the other on the \emph{external} disk $D_0$ has to be attached at distinct points with the same sign. For the other two combinations, the signs have to differ.

The curves are called \emph{strings}, their complement in $D_0^\circ - \mathcal{D}_T$ is the set of \emph{regions}. 
\item[{(4)}] We then put a \emph{shading} on the regions like so: The region on the right of the marked point of the external disk is 
\begin{align*}
	\begin{cases}
		\text{shaded} & \text{if } \epsilon_0 = +,\\
		\text{not shaded} & \text{if } \epsilon_0= -.
	\end{cases}
\end{align*}
All other regions are then either shaded or not shaded such that no two adjacent regions have the same shading, i.e.\ we have a checkerboard shading. This works because of the compatibility condition in point {(3)}.

\end{itemize}

Finally, these objects are defined up to isotopy; that is, if you can continuously deform one tangle into another one then both are considered the same.
\end{definition}
Two enlightening examples can be found in \figref{fig:Tangle Examples1}.
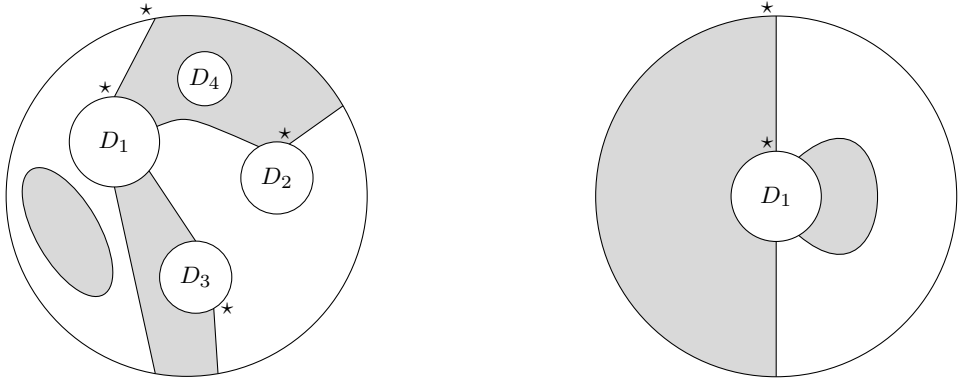
\begin{figure}[!htp]
	\centering
		\begin{minipage}[b]{.4\textwidth}
		\begin{tikzpicture}[scale=1.2]
			\coordinate (centerO) at (0,0);
			\coordinate (centerA) at (-0.8,0.6);
			\coordinate (centerB) at (1,0.2);
			\coordinate (centerC) at (0.1,-0.9);
			\coordinate (centerD) at (0.2,1.3);
			\def\radiusO{2cm};
			\def\radiusA{0.5cm};
			\def\radiusB{0.4cm};
			\def\radiusC{0.4cm};
			\def\radiusD{0.3cm};
			\def\helligkeit{30};
			
			\begin{scope}
				\path[clip] ($(centerO) + (100:\radiusO)$) -- ($ (centerA) +(90:\radiusA) $) --
					($ (centerA) +(90:\radiusA) $) arc [start angle=90 , end angle=20 , radius=\radiusA] --
					($(centerA) + (20:\radiusA)$) .. controls (0,0.9).. ($ (centerB) +(120:\radiusB) $) --
					($ (centerB) +(120:\radiusB) $) arc [start angle=120 , end angle=70 , radius=\radiusB] --
					($ (centerB) +(70:\radiusB) $) -- ($(centerO) + (30:\radiusO)$) --
					($(centerO) + (30:\radiusO)$) arc [start angle=30 , end angle=100 , radius=\radiusO];				
				\fill[color = gray!\helligkeit, opacity=0.1]  (centerO) circle [radius=\radiusO];
			\end{scope}
			
			\begin{scope}			
				\path[clip] ($(centerA) + (-40:\radiusA)$) -- 
					($ (centerC) +(90:\radiusC) $) arc [start angle=90 , end angle=300 , radius=\radiusC] --
					($ (centerC) +(300:\radiusC) $) -- 
					($(centerO) + (-80:\radiusO)$) arc[start angle=-80, end angle=-100, radius=\radiusO] --
					($ (centerA) +(270:\radiusA) $) arc[start angle=270, end angle=320, radius=\radiusA];
				\fill[color = gray!\helligkeit, opacity=0.1]  (centerO) circle [radius=\radiusO];
			\end{scope}
					
			\draw (centerO) circle [radius=\radiusO];
			\draw[fill=white] (centerA) circle [radius=\radiusA];
			\draw[fill=white] (centerB) circle [radius=\radiusB];
			\draw[fill=white] (centerC) circle [radius=\radiusC];
			\draw[fill=white] (centerD) circle [radius=\radiusD];
			
			\draw (centerA) node {$D_1$};
			\draw (centerB) node {$D_2$};
			\draw (centerC) node {$D_3$};
			\draw (centerD) node {$D_4$};	
	
			\draw 
				($(centerO) + (100:\radiusO)$) -- ($ (centerA) +(90:\radiusA) $);
			\draw 
				($(centerO) + (260:\radiusO)$) -- ($ (centerA) +(270:\radiusA) $);
			\draw 
				($(centerO) + (30:\radiusO)$) -- ($ (centerB) +(70:\radiusB) $);
			\draw 
				($(centerO) + (-80:\radiusO)$) -- ($ (centerC) +(300:\radiusC) $);
			
			\draw 
				($(centerA) + (20:\radiusA)$) .. controls (0,0.9).. ($ (centerB) +(120:\radiusB) $);
			\draw 
				($(centerA) + (-40:\radiusA)$) -- ($ (centerC) +(90:\radiusC) $);
			
			\draw[fill=gray!\helligkeit, opacity=0.1] (-1.32,-0.4) circle [x radius = 0.35cm, y radius = 0.8cm, rotate = 30];
			
			\node (markedPointO) at ($(centerO) + (100:\radiusO) + (-0.1,0.1)$) {$\markedPoint$};
			\node (markedPointA) at ($(centerA) + (90:\radiusA) + (-0.1,0.1)$) {$\markedPoint$};
			\node (markedPointB) at ($(centerB) + (70:\radiusB) + (-0.05,0.12)$) {$\markedPoint$};
			\node (markedPointC) at ($(centerC) + (-60:\radiusC) + (0.15,0)$) {$\markedPoint$};
		\end{tikzpicture}
		\end{minipage}\qquad
		\begin{minipage}[b]{.4\textwidth}

		\begin{tikzpicture}[scale=1.2]
			\coordinate (centerO) at (0,0);
			\coordinate (centerA) at (0,0);
			\def\radiusO{2cm};
			\def\radiusA{0.5cm};
			\def\helligkeit{30};
			
			\begin{scope}
				\path[clip]
				($ (centerA) +(90:\radiusA) $) arc[start angle=90, end angle=270, radius=\radiusA] --
				($(centerO) + (-90:\radiusO)$) arc[start angle=270, end angle=90, radius=\radiusO];
				\fill[color = gray!\helligkeit, opacity=0.1]  (centerO) circle [radius=\radiusO];
			\end{scope}
			\begin{scope}
				\path[clip] 
					($(centerA) + (60:\radiusA)$) .. controls ($(centerA) + (45:\radiusO)$) and ($(centerA) + (-45:\radiusO)$) .. ($ (centerA) +(-60:\radiusA) $) arc[start angle=300, end angle=60, radius=\radiusA];
				\fill[color = gray!\helligkeit, opacity=0.1]  (centerO) circle [radius=\radiusO];				
			\end{scope}
					
			\draw (centerO) circle [radius=\radiusO];
			\draw[fill=white] (centerA) circle [radius=\radiusA];
			
			\draw (centerA) node {$D_1$};	
	
			\draw 
				($(centerO) + (90:\radiusO)$) -- ($ (centerA) +(90:\radiusA) $);
			\draw 
				($(centerO) + (-90:\radiusO)$) -- ($ (centerA) +(-90:\radiusA) $);		
			\draw 
				($(centerA) + (60:\radiusA)$) .. controls ($(centerA) + (45:\radiusO)$) and ($(centerA) + (-45:\radiusO)$) .. ($ (centerA) +(-60:\radiusA) $);		
			
			\node (markedPointO) at ($(centerO) + (90:\radiusO) + (-0.1,0.1)$) {$\markedPoint$};
			\node (markedPointA) at ($(centerA) + (90:\radiusA) + (-0.1,0.1)$) {$\markedPoint$};
		\end{tikzpicture}
		\end{minipage}
		\caption{Examples of planar tangles. In the left figure we see a $(4,+)$-tangle whose internal disks have colors $(4,+),(2,-),(2,+),$ and $(0,+)$, respectively. The tangle in figure to the right is $(2,-)$ with a single internal $(4,-)$ disk.}\label{fig:Tangle Examples1}
\end{figure}

Note that the shading is really only a way of encoding the colors of the disk so that they are easier to recognize. It is not necessary to always draw it. The shading is clear as long as one keeps track of the colors and the marked points of each disk.

Given an $(n,\epsilon)$-tangle $T$ and an $(\tilde{n},\tilde{\epsilon})$-tangle $S$ such that for some internal disk $D_i$ of $T$ we have $(n_i, \epsilon_i)=(\tilde{n},\tilde{\epsilon})$ we can in a natural way define a new $(n,\epsilon)$-tangle $T\circ_i S$. Namely: plug $S$ into $D_i$ such that the marked and all other distinct points of $S$ and $D_i$ coincide, delete the boundary of $S$, and smooth out the strings. Then 
\begin{align*}
\mathcal{D}_{T\circ_i S} = \left(\mathcal{D}_T- \{D_i\}\right) \cup \mathcal{D}_S,
\end{align*}
and we relabel the internal disks in the obvious way. If there is only one internal disk in which we can insert then we will also drop the index and write $\circ$ instead of $\circ_1$.
We can, for example, $\circ_2$-compose the first tangle in \figref{fig:Tangle Examples1} with the second one and obtain
\begin{center}
\begin{tikzpicture}[scale=1]
	\coordinate (centerO) at (0,0);
	\coordinate (centerA) at (-0.8,0.6);
	\coordinate (centerB) at (1,0.2);
	\coordinate (centerC) at (0.1,-0.9);
	\coordinate (centerD) at (0.2,1.3);
	\def\radiusO{2cm};
	\def\radiusA{0.5cm};
	\def\radiusB{0.4cm};
	\def\radiusC{0.4cm};
	\def\radiusD{0.3cm};
	\def\helligkeit{30};
	
	\begin{scope}
		\path[clip] ($(centerO) + (100:\radiusO)$) -- ($ (centerA) +(90:\radiusA) $) --
			($ (centerA) +(90:\radiusA) $) arc [start angle=90 , end angle=20 , radius=\radiusA] --
			($(centerA) + (20:\radiusA)$) .. controls (0,0.9).. ($ (centerB) +(120:\radiusB) $) --
			($ (centerB) +(120:\radiusB) $) arc [start angle=120 , end angle=70 , radius=\radiusB] --
			($ (centerB) +(70:\radiusB) $) -- ($(centerO) + (30:\radiusO)$) --
			($(centerO) + (30:\radiusO)$) arc [start angle=30 , end angle=100 , radius=\radiusO];				
		\fill[color = gray!\helligkeit, opacity=0.1]  (centerO) circle [radius=\radiusO];
	\end{scope}
	
	\begin{scope}
		\path[clip]
			($(centerB) + (240:\radiusB)$) 
			.. controls ($(centerB) + (240:\radiusB) + (240:0.5cm)$) and ($(centerB) + (300:\radiusB)+(300:0.5cm)$) .. 
			($ (centerB) +(300:\radiusB) $) arc[start angle=-60, end angle=-120, radius=\radiusB];
		\fill[color = gray!\helligkeit, opacity=0.1]  (centerO) circle [radius=\radiusO];
	\end{scope}
	
	\begin{scope}			
		\path[clip] ($(centerA) + (-40:\radiusA)$) -- 
			($ (centerC) +(90:\radiusC) $) arc [start angle=90 , end angle=300 , radius=\radiusC] --
			($ (centerC) +(300:\radiusC) $) -- 
			($(centerO) + (-80:\radiusO)$) arc[start angle=-80, end angle=-100, radius=\radiusO] --
			($ (centerA) +(270:\radiusA) $) arc[start angle=270, end angle=320, radius=\radiusA];
		\fill[color = gray!\helligkeit, opacity=0.1]  (centerO) circle [radius=\radiusO];
	\end{scope}
			
	\draw (centerO) circle [radius=\radiusO];
	\draw[fill=white] (centerA) circle [radius=\radiusA];
	\draw[fill=white] (centerB) circle [radius=\radiusB];
	\draw[fill=white] (centerC) circle [radius=\radiusC];
	\draw[fill=white] (centerD) circle [radius=\radiusD];
	
	\draw (centerA) node {$D_1$};
	\draw (centerB) node {$D_2$};
	\draw (centerC) node {$D_3$};
	\draw (centerD) node {$D_4$};	

	\draw 
		($(centerO) + (100:\radiusO)$) -- ($ (centerA) +(90:\radiusA) $);
	\draw 
		($(centerO) + (260:\radiusO)$) -- ($ (centerA) +(270:\radiusA) $);
	\draw 
		($(centerO) + (30:\radiusO)$) -- ($ (centerB) +(70:\radiusB) $);
	\draw 
		($(centerO) + (-80:\radiusO)$) -- ($ (centerC) +(300:\radiusC) $);
	
	\draw 
		($(centerA) + (20:\radiusA)$) .. controls (0,0.9).. ($ (centerB) +(120:\radiusB) $);
	\draw 
		($(centerA) + (-40:\radiusA)$) -- ($ (centerC) +(90:\radiusC) $);
	\draw
		($(centerB) + (240:\radiusB)$) .. controls ($(centerB) + (240:\radiusB) + (240:0.5cm)$) and ($(centerB) + (300:\radiusB)+(300:0.5cm)$) .. ($ (centerB) +(300:\radiusB) $);
	
	\draw[fill=gray!\helligkeit, opacity=0.1] (-1.32,-0.4) circle [x radius = 0.35cm, y radius = 0.8cm, rotate = 30];
	
	\node (markedPointO) at ($(centerO) + (100:\radiusO) + (-0.1,0.1)$) {$\markedPoint$};
	\node (markedPointA) at ($(centerA) + (90:\radiusA) + (-0.1,0.1)$) {$\markedPoint$};
	\node (markedPointB) at ($(centerB) + (70:\radiusB) + (-0.05,0.12)$) {$\markedPoint$};
	\node (markedPointC) at ($(centerC) + (-60:\radiusC) + (0.15,0)$) {$\markedPoint$};
\end{tikzpicture}
\end{center}

This operation of inserting planar tangles into like-colored internal discs of another tangle is called \emph{composition}, and with it we have pretty much already defined the \emph{colored planar operad} $\mathfrak{P}$, see e.g.\ \cite{loday_algebraic_2012}. (If you know about operads then you may see the planar operad as a more ``pedantic'' version of the \emph{little disks} operad.) What this really means is that, for all intents and purposes, we have a collection of things -- in this case, all planar tangles -- and on this collection we define a notion of \emph{composition} that is associative and unital, where unitality of composition  means that there are some elements $\id_{(n_i,\epsilon_i)}$ in our collection such that $T\circ_i \id_{(n_i,\epsilon_i)} = T$, whenever this composition is possible.

It is immediately clear what these units are in our case: For fixed $(n,\epsilon)$ they are the unique $(n,\epsilon)$-tangle with exactly one internal $(n,\epsilon)$-disc, and strings only connecting the internal and the external disks such that both marked points are endpoints of the same string. For example, the following is $\id_{(2,-)}$:
\begin{align*}
\id_{(2,-)} =\,
\begin{tikzpicture}[scale=0.8, baseline]
	\coordinate (O) at (0,0);
	\def\radiusO{2cm};
	\def\radiusA{0.5cm};
	\def\helligkeit{30};
%
	\begin{scope}
		\path[clip]
			($ (O) +(120:\radiusO) $) arc[start angle=120, end angle=240, radius=\radiusO] --
			($ (O) +(240:\radiusA) $) arc[start angle=240, end angle=120, radius=\radiusA] ;
		\fill[color = gray!\helligkeit, opacity=0.1]  (O) circle [radius=\radiusO];
	\end{scope}
	\begin{scope}
		\path[clip] 
			($ (O) +(60:\radiusO) $) arc[start angle=60, end angle=-60, radius=\radiusO] --
			($ (O) +(-60:\radiusA) $) arc[start angle=-60, end angle=60, radius=\radiusA] ;
		\fill[color = gray!\helligkeit, opacity=0.1]  (O) circle [radius=\radiusO];				
	\end{scope}
	\draw (O) circle [radius=\radiusO];
	\draw[fill=white] (O) circle [radius=\radiusA];
%
	\draw 
		($(O) + (120:\radiusO)$) -- ($ (O) + (120:\radiusA) $);
	\draw 
		($(O) + (60:\radiusO)$) -- ($ (O) + (60:\radiusA) $);		
	\draw 
		($(O) + (-120:\radiusO)$) -- ($ (O) + (-	120:\radiusA) $);
	\draw 
		($(O) + (-60:\radiusO)$) -- ($ (O) + (-60:\radiusA) $);		
%
	\node (markedPointO) at ($(O) + (120:\radiusO) + (-0.1,0.1)$) {$\markedPoint$};
	\node (markedPointA) at ($(O) + (120:\radiusA) + (-0.2,0.1)$) {$\markedPoint$};	
\end{tikzpicture}
\end{align*}
As the symbol already implies it also serves as the identity, since for any $(2,-)$-tangle $T$ we clearly have $\id_{(2,-)} \circ T = T$. 

For $x\in\mathfrak{Col}$ we define $\mathcal{T}_x$ to be the free complex vector space on the set of all planar $x$-tangles. These spaces are often called \emph{$x$-box spaces} \footnote{Strictly speaking this name only makes sense for shaded tangles.}, as explained in the following.
Note that $\mathcal{T}_{(n,+)}\cong \mathcal{T}_{(n,-)}$ -- every tangle can be defined with the opposite shading. 

We promised that there is also a description of tangles in terms of rectangles (or \emph{boxes}) instead of discs. This is sometimes easier, and it surely is easier to draw. To obtain it, simply replace all disks by rectangles (or blow up the disks into rectangles, if you will) so that half of the distinct points is along the top edge and the other half is along the bottom edge, and so that on left of the marked point is a corner of the rectangle.

Then we further impose that in this description the marked point be the left-most point on the top. This can always be done by choosing an appropriate representative of the isotopy class, and it is commonly referred to as the \emph{standard form} of the tangle. It allows us to draw pictures even simpler: the symbol $\markedPoint$ doesn't need to be put next to the marked points anymore, as we now know exactly where they are.

In terms of boxes it is now easier to talk about an involution that exists on planar tangles. This involution ${}^*$ yields the \emph{adjoint} $T^*$ of a tangle $T$ as follows.
\begin{itemize}
\item[{$\bullet$}] Take the tangle in standard form.
\item[$\bullet$] Flip it horizontally, i.e.\ along the horizontal axis.
\item[$\bullet$] Interpret the result as a tangle in standard form -- that is, the point that was the ``last'' point (when enumerating in a clockwise fashion) before the flip is now the marked point.
\item[$\bullet$] Note that this operation reversed the shading. But we want the adjoint to have the same color, so invert the shading.
\end{itemize}
That this is indeed involutive, i.e.\ ${}^{**}=\id$, is an easy exercise.

Before we go on, let us list some planar tangles that we will encounter a few times.

\subsubsection*{A few important tangles}
There are a few tangles that deserve to be mentioned explicitly. All of these tangles can be defined for any box space, so we will only give a few easy examples. Arguably the most important one, the \emph{multiplication tangle}, is given by stacking tangles on top of each other, na\"ively speaking. Its pictorial representation can be found after the definition of planar algebras.

The (second) most important one, which will be used throughout this Letter, is the \emph{rotation tangle} or \emph{one-click rotation} $\text{rot}:\mathcal{T}_{(n,\pm)}\rightarrow \mathcal{T}_{(n,\mp)}$. An example can be found in \figref{fig:rotation tangle}. Note that what makes it a ``rotation'' is that the marked point of the internal box is offset by one with respect to the marked point of the exterior box.

\begin{figure}[!htp]\centering
	\begin{tikzpicture}
		\node (text) at (-3.3,0) {$\text{rot}:\mathcal{T}_{(2,+)}\rightarrow\mathcal{T}_{(2,-)} \equiv$};
		\def\helligkeit{30};
		\begin{scope}
			\path[clip]
				(-0.2,0.5) arc[start angle=0, end angle=180, radius=3mm] -- ++(0,-1.7) -- ++(-0.2,0) -- ++(0,2.4) -- ++(1.2,0) -- ++(0,-0.7);
			\fill[color = gray!\helligkeit, opacity=0.1]  (-1,-1.2) rectangle (1,1.2);
		\end{scope}
		\begin{scope}
			\path[clip]
				(0.2,-0.5) arc[start angle=180, end angle=360, radius=3mm] -- ++(0, 1.7) -- ++(0.2,0) -- ++(0,-2.4) -- ++(-1.2,0) -- ++(0,0.7);
			\fill[color = gray!\helligkeit, opacity=0.1]  (-1,-1.2) rectangle (1,1.2);
		\end{scope}
		
		\draw[fill= white] (-0.5,-0.5) rectangle (0.5,0.5);
		\node (O) at (0,0) {$D_1$};
		\draw (-1,-1.2) rectangle (1,1.2);
		\draw (-0.2,0.5) arc[start angle=0, end angle=180, radius=3mm] -- ++(0,-1.7);
		\draw (0.2,-0.5) arc[start angle=180, end angle=360, radius=3mm] -- ++(0,1.7);
		\draw (0.2, 0.5) -- ++(0,0.7);
		\draw (-0.2, -0.5) -- ++(0,-0.7);
	\end{tikzpicture}
	\caption[The rotation tangle]{The shaded rotation tangle for $(2,+)$-tangles. All boxes are in standard form. The shading is supplied for clarity.}
	\label{fig:rotation tangle}
\end{figure}
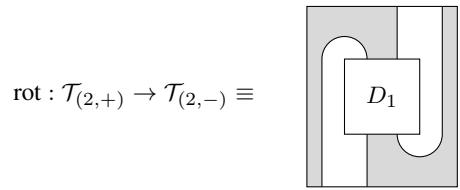
Now the notation will be getting a bit messy, but we promise that the reader only has to endure this for a short while. We can also define the \emph{inclusion tangle} $\text{inc}_{k,\pm}^{k+1}$ (see \figref{fig:InclusionTangle}), which takes a $(k,\pm)$-tangle $S$ and maps it to the $(k+1,\pm)$-tangle $\text{inc}_{k,\pm}^{k+1}\circ_1 S$. It of course also makes sense to then define the general symbol $\text{inc}_{k,\pm}^{l}\equiv \text{inc}_{l-1,\pm}^{l}\circ \ldots \circ \text{inc}_{k,\pm}^{k+1}$ for $l>k$, this is just including the inclusion.

There also exists a sort of converse for the inclusion, commonly called the \emph{conditional expectations tangle} $\mathcal{E}_{k+1,\pm}^{k}$ as seen in \figref{fig:ConditionalExpectationTangle}, taking $(k+1,\pm)$-tangles to $k$-tangles with the same shading for the marked point. For this tangle an obvious generalization $\mathcal{E}_{k+1,\pm}^{l}$, for $l\leq k$, is readily defined.

\begin{figure}[!htp]\centering
	\begin{minipage}[b]{.4\textwidth}\centering
		\begin{tikzpicture}[scale=0.8, every node/.style={scale=0.8}]
			\coordinate (centerO) at (0,0);
			\def\radiusO {2cm};
			\def\radiusA {.7cm};
		
			\begin{scope}		
				\path[clip] ($(centerO)+(90:\radiusO)$) -- 
					($(centerO)+(90:\radiusA)$)  arc [start angle=90, end angle=-90, radius=\radiusA] --
					($(centerO) + (-90:\radiusO)$) --
					($(centerO)+(-90:\radiusO)$) arc [start angle=270, end angle=315, radius=\radiusO] --
					($(centerO)+(-45:\radiusO)$) -- 
					($(centerO) + (45:\radiusO)$) arc [start angle=45, end angle=90, radius=\radiusO];	
				\fill[color = gray!30, opacity=0.1]  (centerO) circle [radius=\radiusO];
			\end{scope}
			
			\draw (centerO) circle [radius=\radiusO];
			\draw (centerO) circle [radius=\radiusA];
			\draw ($(centerO)+(90:\radiusO)$) -- ($(centerO)+(90:\radiusA)$);
			\draw ($(centerO)+(-90:\radiusO)$) -- ($(centerO)+(-90:\radiusA)$);
			\draw ($(centerO)+(45:\radiusO)$) -- ($(centerO) + (-45:\radiusO)$);
			\draw (centerO) node {$D_1$};
			
			\draw ($(centerO)+(90:\radiusO) + (-0.1,0.1)$) node {$\markedPoint$};
			\draw ($(centerO)+(90:\radiusA) + (-0.1,0.1)$) node {$\markedPoint$};
		\end{tikzpicture}
		\caption[]{}
		\label{fig:InclusionTangle}
	\end{minipage}
	\begin{minipage}[b]{.4\textwidth}\centering
		\begin{tikzpicture}[scale=0.8, every node/.style={scale=0.8}]
			\coordinate (centerO) at (0,0);
			\def\radiusO {2cm};
			\def\radiusA {.7cm};		
		
			\begin{scope}		
				\path[clip] ($(centerO)+(90:\radiusO)$) -- 
					($(centerO)+(90:\radiusA)$)  arc [start angle=90, end angle=45, radius=\radiusA] --
					($(centerO)+(45:\radiusA)$) .. 
						controls ($(centerO)+(40:1.6cm)$) and ($(centerO)+(0:1.2cm)$) ..
						($(centerO)+(0:1.2cm)$)..
						controls ($(centerO)+(0:1.2cm)$) and ($(centerO)+(-40:1.6cm)$) ..
						($(centerO)+(-45:\radiusA)$) --
					($(centerO)+(-45:\radiusA)$) arc [start angle=-45, end angle=-90, radius=\radiusA] --
					($(centerO)+(-90:\radiusO)$) arc [start angle=-90, end angle=90, radius=\radiusO];	
				\fill[color = gray!30, opacity=0.1]  (centerO) circle [radius=\radiusO];
			\end{scope}
			
			\draw (centerO) circle [radius=\radiusO];
			\draw (centerO) circle [radius=\radiusA];
			\draw ($(centerO)+(90:\radiusO)$) -- ($(centerO)+(90:\radiusA)$);
			\draw ($(centerO)+(-90:\radiusO)$) -- ($(centerO)+(-90:\radiusA)$);
			\draw ($(centerO)+(45:\radiusA)$) .. 
				controls ($(centerO)+(40:1.6cm)$) and ($(centerO)+(0:1.2cm)$) ..
				($(centerO)+(0:1.2cm)$)..
				controls ($(centerO)+(0:1.2cm)$) and ($(centerO)+(-40:1.6cm)$) ..
				($(centerO)+(-45:\radiusA)$);

			\draw (centerO) node {$D_1$};
			
			\draw ($(centerO)+(90:\radiusO) + (-0.1,0.1)$) node {$\markedPoint$};
			\draw ($(centerO)+(90:\radiusA) + (-0.1,0.1)$) node {$\markedPoint$};
		\end{tikzpicture}
		\caption[]{}
		\label{fig:ConditionalExpectationTangle}
	\end{minipage}
	
	\caption[Inclusion and conditional expectation tangles]{The inclusion tangle $\text{inc}_{1,+}^{2}:\mathcal{T}_{(1,+)}\rightarrow\mathcal{T}_{(2,+)}$ and the conditional expectation tangle $\mathcal{E}_{2,+}^{1}$.
	}
\end{figure}

Finally, we can also define another conditional expectation tangle that closes to the left. But observe that this one will switch the shading. We then get two a priori distinct notions of what we call  \emph{trace}, by fixing one of the two ways of building conditional expectation tangles (i.e.\ closing left or right), and successively applying it until the result lives in either $\mathcal{T}_{(0,+)}$ or $\mathcal{T}_{(0,-)}$.

Before going on to defining planar algebras, let's briefly discuss \emph{unshaded tangles}. The name already suggests the main difference: they do not come with a shading. Without a shading, an odd number of boundary points should be possible as well, so we say that the set of colors is then the set of natural numbers, i.e.\ $\mathfrak{Col} =\mathbb{N}_0$, and we consequently require an unshaded $n$-tangle to have $n$ distinct points on the boundary -- in the shaded case we needed an even number of points so that shading is actually possible. It is clear that then, in general, a representation in terms of boxes is not possible -- only when the number of points is even. 
\subsubsection{Planar Algebras}
Here, again, we will discuss \emph{shaded planar algebras} first, after which we briefly mention the unshaded case. In the fewest possible words,
\begin{center}
\begin{minipage}{0.8\textwidth}
A shaded planar algebra is an operad homomorphism from $\mathfrak{P}$ to a colored operad $\mathfrak{V}$ of vector spaces.
\end{minipage}
\end{center}
That sounds weird, so let us unravel it a bit. Consider a collection $P\equiv\left\{ P_x \right\}_{x\in\mathfrak{C}}$ of vector spaces, indexed by a set of colors. This collection lives in the symmetric monoidal category $\catname{Vect}_k$ of vector spaces over some field $k$, and we can look at things like
\begin{align*}
\mathrm{Hom}\left( \bigotimes_{x\in U} P_x, P_y\right),\quad \text{for } U\subset\mathfrak{Col},
\end{align*}
i.e.\ linear maps from the tensor product of some of the vector spaces to another vector space in $P$. To discover the operadic structure $\mathfrak{V}$ coming with $P$, we make use of the so-called \emph{tensor-hom adjunction}, which tells us that there is an isomorphism
\begin{align*}
\mathrm{Hom}\left( A\otimes B, C \right)\cong \mathrm{Hom}\left( A,\mathrm{Hom}\left( B,C \right) \right),
\end{align*}
natural in $A, C$ for every $B$ -- basically currying. Indeed, the isomorphism may be viewed as
\begin{align*}
A\otimes B \rightarrow C \overset{\sim}{\longleftrightarrow} A\rightarrow\left( B\rightarrow C \right).
\end{align*}

Let $U,V\subset \mathfrak{Col}$ be finite, and $x,y\in\mathfrak{Col}$. For any two maps
\begin{align*}
\bigotimes_{a_i\in U} P_{a_i}\morphism{f} P_x \qquad\text{and}\qquad \bigotimes_{b_j\in V} P_{b_j}\morphism{g} P_y,
\end{align*}
there exists a natural composition whenever $x=b_j$. To see this, fix some $l$ such that $x=b_l$.
By the tensor-hom adjunction there exists a map
\begin{align*}
P_x = P_{b_l}  \overset{\tilde{g}}{\dashrightarrow} \mathrm{Hom}\left( \bigotimes_{b\in V-\{b_l\}} P_b, P_y \right)
\end{align*}
corresponding uniquely to $g$. Then we define the composition $g\diamond_l f$ to be the map corresponding to
\begin{align*}
\tilde{g} \circ f \in& \,\mathrm{Hom}\left( \bigotimes_{a_i\in U} P_{a_i}, \mathrm{Hom}\left( \bigotimes_{b\in V-\{b_l\}} P_b, P_y  \right) \right)\\[2em]
\cong& \, \mathrm{Hom}\left(  \bigotimes_{a \in U\cup V-\{b_l\}} P_{a}, P_y\right)
\end{align*}
via the tensor-hom adjunction.
It is now easy to think about what would be considered an ``operad homomorphism'', and we record this in the following definition.
\begin{definition}[Planar Algebra]\index{Planar Algbra}\label{def:Planar Algebra}
The collection $P$ becomes a planar algebra by choosing an interpretation $Z:\mathfrak{P}\rightarrow\mathfrak{V}$ of planar tangles as linear maps, compatible with the operadic structures.

That is, an $(n,\epsilon)$-tangle $T$ with $k$ internal disks will be mapped to the linear transformation
\begin{align*}
Z(T): \bigotimes_{i=1}^k P_{(n_i,\epsilon_i)} \rightarrow P_{(n,\epsilon)},
\end{align*}
i.e.\ a linear map from the spaces associated with the inner boundaries of the tangle to the space associated with the outer boundary. If $k=0$ then the empty tensor product is the underlying field, and $Z(T)$ is therefore a vector in $P_{(n,\epsilon)}$. We also impose $Z(\id_{(n,\epsilon)}) = \id_{(n,\epsilon)}$.

The compatibility is then that $Z$ is something akin to a homomorphism, i.e.\
\begin{align*}
Z(T\circ_i S) = Z(T)\diamond_i Z(S)
\end{align*}
We will also call $P$ itself a planar algebra, and use $Z$ to denote the homomorphism in general.
\end{definition}

We have already implicitly encountered one example. If $P_x = \mathcal{T}_x$ for $x\in\mathfrak{Col}$, then we get an obvious planar algebra structure and can easily see that each $\mathcal{T}_{(n,\epsilon)}$ is turned into an associative algebra with multiplication
\begin{align*}\mathrm{mult}_n \equiv \,
\begin{tikzpicture}[scale=0.9, baseline]
	\coordinate (O) at (0,0);
	\node[draw] (A) at (0,0.8) {$D_2$}; 
	\node[draw] (B) at (0,-0.8) {$D_1$};
	\def\radiusO{2cm};
	\def\radiusA{.6cm};
	\draw (-1,-2) rectangle (1,2);
	\draw (0,2) -- node[right] {$n$} (A.north);
	\draw (A.south) -- node[right] {$n$} (B.north);
	\draw (B.south) -- node[right] {$n$} (0,-2);
\end{tikzpicture}\,.
\end{align*}
Here and in the following, an $n$ next to a string implies that there are actually $n$ parallel copies of that string, and the marked points are at the top left of the respective box, as discussed previously. The shading was not drawn, because it is implicitly clear.

The associativity of the multiplication is easily seen by simply remembering  that tangles are defined up to isotopy so that 
\[ \mult_n \circ_1 \mult_n \sim \mult_n\circ_2 \mult_n. \]
Note that this is even commutative if $n=0.$

\subsubsection*{Additional stuff, structure, property}
So far the family of planar $(n,\epsilon)$-tangles is quite large, because there can be an arbitrary collection of closed curves on the inside. Usually, the planar algebras one encounters have an additional property regarding these.

\begin{definition}[Modulus]
Let $P$ be a planar algebra over $\mathbb{F}$. Let $T$ be a tangle with a closed curve, and let $\hat{T}$ be the same tangle without that curve.

If there exists a
\begin{align*}
\delta \equiv 
\begin{cases}
	\delta_+ &\text{if region bounded by curve is shaded} \\
	\delta_- &\text{if region bounded by curve is unshaded}	
\end{cases}
\in\mathbb{F}
\end{align*}
such that
\begin{align*}
Z(T) = \delta Z(\hat{T})
\end{align*}
then $P$ is said to have \emph{shaded (unshaded) modulus} $\delta_+$ ($\delta_-$). We will only ever encounter the case where $\delta_+ = \delta_->0$, and we then simply say that \emph{$P$ has (positive) modulus $\delta$.}
\end{definition}
\noindent Looking forward to the Temperley-Lieb algebras, $\delta$ is also called \emph{loop value}.
We may also want to carry the involution on tangles to planar algebras.
\begin{definition}[Planar ${}^*$-algebra]
A planar algebra $P$ over $\mathbb{F}$ is called a \emph{planar ${}^*$-algebra} if $\mathbb{F}$ has a conjugation and if for all colors $x\in\mathfrak{Col}$ the vector space $P_x$ possesses an antilinear involution, such that this involution commutes with the one on planar tangles, that is
\begin{align*}
Z_T(v_1, \ldots, v_k)^* = Z_{T^*}(v_1^*,\ldots, v_k^*)
\end{align*}
for an $n$-tangle $T$ with $k$ internal boxes.
\end{definition}
If the vector spaces are even $C^*$-algebras with that involution, then one speaks of a \emph{planar $C^*$-algebra}.  We only encounter planar algebras consisting of complex vector spaces in this Letter.

The motivation for planar algebras came from the study of subfactors. As proven by Jones \cite{jones_planar_1999} subfactors give rise to a certain type of planar algebra, and every planar algebra of that type comes from a subfactor. We now define these important PAs. 
\begin{definition}[Subfactor Planar Algebra]\index{Planar Algebra!Subfactor}
A planar ${}^*$-algebra $S$ is a \emph{subfactor planar algebra} if
\begin{itemize}
\item[{(i)}] All vector spaces are finite-dimensional, and $\dim S_{0,\pm}=1$. This last property is sometimes called \emph{connectedness}.
\item[{(ii)}] It is \emph{spherical}: The assignment of complex numbers to labelled \footnote{If $T$ is a $(0, \epsilon_0)$-tangle with internal disks $D_i$, then a \emph{labelling} is an element in the range of $Z_T$, drawn as $T$ with some vectors $v_i\in S_{(k_i, \epsilon_i)}$ in the preimage `inserted' into its $i$th internal disk. A labelling is thus equivalently a choice of elements in the domain of $Z_T$.}
 $(0,\pm)$-tangles is invariant under isotopies of the 2-sphere. This assignment is called \emph{partition function}. Sphericality means that left and right traces coincide.
\item[{(iii)}]We can define a positive-definite bilinear form $\langle v,w \rangle\equiv \mathrm{tr}(w^* v)$, which is then called the \emph{inner product}.
\end{itemize}
\end{definition}

A subfactor planar algebra thus has single modulus, as implied by items {(i)} and {(ii)}.

\subsection{The Temperley-Lieb Planar Algebra}
The Temperley-Lieb planar algebra $\mathfrak{TL}$ is as simple as it gets, and we can define it as a shaded PA with modulus. The vector spaces consist of Temperley-Lieb diagrams, i.e.\ they are the subspaces of $\mathcal{T}_x$ consisting of all diagrams without internal boxes. It is then clear how to put the two different shadings on a tangle, by simply choosing the color of the first region after the marked point. That way we already get all vector spaces $TL_{(n,\epsilon)}$ whose colors are in $\mathbb{N}\times\{+,-\}$. The two remaining vector spaces $TL_{0,\pm}$ are identified with $\mathbb{C}$.

Tangles with internal boxes act on this collection of vector spaces in an obvious way: TL basis diagrams are tangles, so we just compose and extend by linearity, for each internal disk. Thus by labelling a planar tangle with vectors in TL algebras we get a (multi-)linear map between (products of) these algebras.

Looking back at the definition of subfactor planar algebra, we see that the TL PA almost satisfies it. It has a natural ${}^*$-structure, and the first two items in the definition follow immediately. The inner product, however, need not exist: the bilinear form defined by the trace is not necessarily positive definite, depending on the loop value. Taking quotients by the elements $x$ for which $\langle x,x \rangle=0$ then yields vectors spaces with honest inner products. We can thus turn every Temperley-Lieb planar algebra into a subfactor PA.

We paraphrase \cite{peters_planar_2010} almost verbatim:
\begin{fact}If $S$ is any subfactor planar algebra with modulus $\delta$, then the obvious `embedding' of the TL PA, i.e.\ the planar algebra homomorphism
\begin{align*}
\mathfrak{TL}\hookrightarrow S,
\end{align*}
given by interpreting Temperley-Lieb tangles as tangles without input disks, is injective if $\delta \geq 2$, otherwise its kernel is the radical of the bilinear form, so that modding out the kernel gives an injective planar algebra homomorphism.
\end{fact}
Surely then a \emph{perfect tangle} $T$ is not in the kernel of that map: It are defined such that $\langle T,T \rangle \propto \mathrm{tr}(\mathbf{1_k})$, where $\mathbf{1}_k$ is the tangle that is the multiplicative unit for the multiplication of $k$-tangles.

\subsection*{Trivalent Categories}
Trivalent categories, as defined in \cite{morrison_categories_2017}, have the nice property that the calculations within are basically calculations in a certain unshaded planar algebra with additional constraints.

We first write down their category-theoretic definition, but from then on only use string diagrams to apply the methods from the theory of planar algebras. In the vein of the seminal paper we first introduce a few preliminary notions.

The first one is \emph{evaluability}: a monoidal $\mathbb{F}$-linear category is \emph{evaluable} if and only if $\dim\mathrm{Hom}(1,1) = 1$, and that endomorphism space is identified with $\mathbb{F}$ by sending the empty string diagram to $1\in \mathbb{F}$.

Secondly, to rule out a certain degeneracy like before, we call a pivotal category \emph{non-degenerate} iff for every nonzero morphism $f:X\rightarrow Y$ we can find a morphism $f^\prime:Y\rightarrow X$ such that the trace $\mathrm{tr}\left( f^\prime\circ f \right)\in\mathrm{Hom}(1,1)$ is not zero.

For a fixed object $X$ we introduce the notation 
\begin{align*}
\mathcal{C}_n \equiv\mathrm{Hom}\left( 1, X^{\otimes n} \right).
\end{align*}
These vector spaces form an unshaded planar algebra, so far only formally. After the next definition we will understand that $\mathcal{C}_n$ is the span of open planar trivalent graphs with $n$ boundary points, up to isotopy rel boundary, so that only tangles with internal disks of color at most $3$ play any role at all.

\begin{definition}[Trivalent Category]
A triple $\left( \mathcal{C},X,\tau \right)$ forms a \emph{trivalent category} if
\begin{itemize}
\item $\mathcal{C}$ is a strict pivotal $\mathbb{C}$-linear category that is evaluable and non-degenerate;
\item $X$ is a self-dual simple object such that 
	\begin{align*}
		\dim\mathcal{C}_1 = 0,\qquad \dim\mathcal{C}_2 = 1,\qquad \dim\mathcal{C}_3 = 1;
    \end{align*}
\item $\tau \in \mathcal{C}_3$, called the \emph{trivalent vertex}, is rotationally invariant, that is 
	\begin{align*}
		\begin{tikzpicture}[scale=0.5, baseline]
			\node[draw] (t) at (0,0) {$\tau$};
			\foreach \x in {-0.3,0} \draw ($(t.north) + (\x,0)$) .. controls +(\x,0.5) .. ++(\x,1);
			\draw ($(t.north) + (0.3,0)$) arc[start angle=180, end angle=0, radius=5mm] --
				++ (0,-1) arc[start angle=0, end angle=-180, radius=1.3cm] -- ++ (0,2);
		\end{tikzpicture}
		\, = \,
			\begin{tikzpicture}[scale=0.5, baseline]
			\node[draw] (t) at (0,0) {$\tau$};
			\foreach \x in {-0.3,0,0.3} \draw ($(t.north) + (\x,0)$) .. controls +(\x,0.5) .. ++(\x,1);
		\end{tikzpicture}\,,
	\end{align*}
	and $\mathcal{C}$ is generated by $\tau$. This means that every morphism is obtained from tensoring and adding identities, (co)evaluations, the self-duality, and trivalent vertices.
\end{itemize}
\end{definition}
Let from now on $\mathcal{C}$ be as in the definition. The object $X$ is actually symmetrically self-dual, meaning that there exists an isomorphism $\phi:X\morphism{\sim} X^*$ with $\phi^* = \phi$. This follows from the non-degeneracy and the dimension of $\mathcal{C}_3$, see Lemma 2.2 in \cite{morrison_categories_2017}.

 It is customary to drop the label of the morphism $\tau$ and to simply denote it by a trivalent vertex. Then every string diagram will be an open planar trivalent graph, where strings correspond to simple objects, and two graphs isotopic rel boundary correspond to the same morphism (the self-duality is not drawn). Our category is \emph{defined} to behave nicely, so we can use planar isotopy to deform graphs. For us it is not necessary to reinterpret the graphs as morphisms, since we are only using the graphical calculus.

 Because of the non-degeneracy and the dimensional restrictions we can already extract a few necessary parameters for trivalent categories. For one, any diagram containing the \emph{tadpole} or \emph{lollipop}
\begin{align*}
\begin{tikzpicture}[scale=0.5]
	\draw (1,0) arc[start angle=0, end angle=360, radius=1cm];
	\draw (0,-2) -- (0,-1);
\end{tikzpicture}
\end{align*}
as a subgraph is zero, since $\dim\mathcal{C}_1=0$.

Next, a loop (a circle shape) is an element in $\mathbb{C}$, called $d$.  This value must be nonzero, because the loop is the basis element of $\mathcal{C}_2$ composed with its dual.

On the other hand, the following element is in $\mathcal{C}_2$ and must thus be a scalar multiple of the single string,
\begin{align*}
\begin{tikzpicture}[scale=1, baseline]
	\node (A) at (-0.5,-0.8) {};
	\node (B) at (0.5,-0.8) {};
	\foreach \node in {A, B}
		\foreach \x in {-0.3,0,0.3} 
			\draw (\node.north) .. controls +(\x,0.5) .. ++(\x,1);
	\draw ($(A.north) + (0.3,1)$) arc[start angle=180, end angle=0, radius=2mm];
	\draw ($(A.north) + (0,1)$) arc[start angle=180, end angle=0, radius=5mm];
	\draw ($(A.north) + (-0.3,1)$) -- ++(0,0.5);
	\draw ($(B.north) + (0.3,1)$) -- ++(0,0.5);
\end{tikzpicture} \, = \, b\cdot \,
\begin{tikzpicture}[baseline]
	\draw (-0.5,0.3) arc[start angle=180, end angle=360, radius=5mm];
\end{tikzpicture}\,,
\end{align*}
and like before putting a cap on this must yield $bd\neq 0$, so the parameter $b$  must be nonzero. The trivalent vertex can be normalized so that we always take $b=1$.

Finally, we can rotate and compose trivalent vertices such that the result lives again in $\mathcal{C}_3$ (just imagine all legs bent up), and it must thus be some multiple of the trivalent vertex itself, like so:
\begin{align*}
\begin{tikzpicture}[baseline]
	\coordinate (a) at (0,0.5);
	\coordinate (b) at (-0.5,-0.2);
	\coordinate (c) at (0.5,-0.2);
	\draw (a) -- ++ (0,0.3);
	\draw (b) -- ++(210:3mm);
	\draw (c) -- ++(330:3mm);
	\path[draw] (b) -- (a) -- (c) -- (b);
\end{tikzpicture}\, = t\cdot \,
\begin{tikzpicture}[baseline]
	\foreach \angle in {90, 210, 330}
		\draw (0,0) -- ++(\angle:5mm);
\end{tikzpicture}.
\end{align*}
However, coming from the composition, $t$ could actually vanish here.

 In their paper, Morrison \emph{et al} then started classifying trivalent categories by
\begin{itemize}
\item[\text{a)}] the dimensions of the morphism spaces $\mathcal{C}_n$, together with
\item[\text{b)}] the parameters $d$ and $t$, often satisfying some polynomial unique to the category.
\end{itemize}

In this Letter, we only look at trivalent categories where $\mathcal{C}_4$ is 4-dimensional, and a basis is provided by the set of all planar trivalent graph with 4 four boundary points and no internal faces. Any trivalent category with this property is called a \emph{cubic category}. Cubic categories have a nice relation for the square (the trivalent graph with four boundary points, four trivalent vertices, and one internal face), and that makes them particularly nice to work with.

It is clear how to interpret open trivalent graphs with $n$-boundary points as unshaded $n$-tangles: We only need to agree on the marked points, and blow up each trivalent vertex to an internal $3$-disk, labelled with $\tau$, and we also simply call that disk $\tau$. \emph{Rotational invariance} then means $\text{rot}\,\tau = \tau$, so the choice of marked point on $3$-disks doesn't actually matter at all.

A caveat is appropriate here: In planar algebras, we used the symbol ${}^*$ do denote the adjoint or dual of a tangle. That symbol, however, is already reserved in trivalent categories for the category-theoretic dual. Thus for a morphism $f$ in a trivalent category, $f^*$ means ``rotate through $\pi$''. Luckily, it is quickly checked that the notion of horizontally reflecting morphisms is just an instance of what has been called \emph{conjugation} before, for example in \cite{wang_topological_2010}, where conjugation of morphisms $f,g$ in a pivotal $\mathbb{C}$-linear category was defined as an anti-linear operation $\overline{\phantom{f}}$ that switches source and target of a morphism and satisfies
\begin{align*}
\overline{\overline{f}} = f,\qquad\qquad \overline{f\otimes g} = \overline{f}\otimes\overline{g},\qquad\qquad \overline{f\circ g} = \overline{g}\circ \overline{f}.
\end{align*}
Keeping that in mind, translating the discussion of \emph{perfect tangles} to trivalent categories is then but a quick change in notation.

\subsubsection{The Temperley-Lieb algebras}
In this appendix we give one of the most important preliminary definitions. We will first define what the Temperley-Lieb algebra is rather abstractly, but quickly give an enlightening example. 

In the following an $n$-box is a rectangle with $n$ distinguished points along the top (w.r.t.\ the $y$-coordinate) boundary, and $n$ distinguished points along the bottom boundary. The size of the rectangle doesn't matter.
\begin{definition}[Temperley-Lieb algebra]\label{def:TLalgebra}\index{Temperley-Lieb algebra}
Let $q\in\mathbb{R}^\times, n\in\mathbb{N}$. The \emph{Temperley-Lieb algebra} $TL_n(q)$ on $n$ \emph{strands} and with \emph{loop parameter} $q$ is the associative algebra built like this:
\begin{itemize}
\item[•] The underlying vector space is the free vector space $\mathbb{C}\mathcal{B}_n$, with
	\begin{align*}
		\mathcal{B}_n \equiv
		\left\{
		\begin{minipage}{0.6\textwidth}
			all possible crossingless matchings of the $2n$ distinguished points on the boundary of an $n$-box
		\end{minipage}
		\right\}\Bigg/\sim	\end{align*}
	By a \emph{matching} we mean a (smooth) curve connecting two distinct points. The meaning of \emph{crossingless} is then clear. Finally, $\sim$ is the equivalence relation of two matchings being isotopic rel boundary.
\item[•] Multiplication $T\cdot S$ for $T,S\in \mathcal{B}_n$ is given by \emph{stacking} the box $S$ on top of $T$, removing the middle line, and then substituting any loop, i.e.\ a closed curve, by a factor of $q$. This is then extended linearly to all of $TL_n(q)$.
\end{itemize}
\end{definition}
Often it is also required that $q$ be positive. We make this definition tangible in the following example.
\begin{example}\label{ex:TL3}
Let $n=3$, $q$ arbitrary. Then the basis of $TL_3(q)$ is the set
\begin{align*}
\mathcal{B}_3 = \left\{\,\,
	\begin{tikzpicture}[baseline=-1mm,scale=0.9]
		\pgfmathsetmacro{\yMax}{0.8};
		\pgfmathsetmacro{\offset}{2.5};
		\pgfmathsetmacro{\comma}{1};
		\foreach \x in {0,...,4}
			\draw[dashed] (-0.7+\offset*\x,-\yMax) rectangle (0.7+\offset*\x,\yMax);
		\begin{scope}[shift={(0,0)}]
			\foreach \x in {-0.5,0,0.5}
				\draw (\x, \yMax) -- (\x,-\yMax);
			\node (comma) at (+\comma,-\yMax) {$,$};
		\end{scope}
		\begin{scope}[shift={(1*\offset,0)}]
			\draw (0,\yMax) .. controls (0,-0.5+\yMax) and (0.5,-0.5+\yMax) .. (0.5,\yMax);
			\draw (-0.5,-\yMax) .. controls (-0.5,1-0.5-\yMax) and (0,1-0.5-\yMax) .. (0,-\yMax);
			\draw (-0.5,\yMax) .. controls (-0.5,0) and (0.5,0) .. (0.5, -\yMax);
			\node (comma) at (+\comma,-\yMax) {$,$};
		\end{scope}
		\begin{scope}[shift={(2*\offset,0)}]
			\draw (-0.5,\yMax) .. controls (-0.5,-0.5+\yMax) and (0,-0.5+\yMax) .. (0,\yMax);
			\draw (-0,-\yMax) .. controls (-0,1-0.5-\yMax) and (0.5,1-0.5-\yMax) .. (0.5,-\yMax);
			\draw (-0.5,-\yMax) .. controls (-0.5,0) and (0.5,0) .. (0.5, \yMax);
			\node (comma) at (+\comma,-\yMax) {$,$};
		\end{scope}
		\begin{scope}[shift={(3*\offset,0)}]
			\draw (-0.5,\yMax) -- (-0.5,-\yMax);
			\draw (0,\yMax) .. controls (0,-0.5+\yMax) and (0.5,-0.5+\yMax) .. (0.5,\yMax);
			\draw (-0,-\yMax) .. controls (-0,1-0.5-\yMax) and (0.5,1-0.5-\yMax) .. (0.5,-\yMax);
			\node (comma) at (+\comma,-\yMax) {$,$};
		\end{scope}
		\begin{scope}[shift={(4*\offset,0)}]
			\draw (0.5,\yMax) -- (0.5,-\yMax);
			\draw (-0.5,\yMax) .. controls (-0.5,-0.5+\yMax) and (0,-0.5+\yMax) .. (0,\yMax);
			\draw (-0.5,-\yMax) .. controls (-0.5,1-0.5-\yMax) and (0,1-0.5-\yMax) .. (0,-\yMax);
		\end{scope}
	\end{tikzpicture}
\,\,\right\}.
\end{align*}
We show how to perform multiplication:
\begin{align*}
\begin{tikzpicture}[scale=0.8]
	\pgfmathsetmacro{\yMax}{0.8};
	\begin{scope}[shift={(0,0)}]
		\draw[dashed] (-0.7,-\yMax) rectangle (0.7,\yMax);
		\draw (0,\yMax) .. controls (0,-0.5+\yMax) and (0.5,-0.5+\yMax) .. (0.5,\yMax);
		\draw (-0.5,-\yMax) .. controls (-0.5,1-0.5-\yMax) and (0,1-0.5-\yMax) .. (0,-\yMax);
		\draw (-0.5,\yMax) .. controls (-0.5,0) and (0.5,0) .. (0.5, -\yMax);
		\node (comma) at (1,0) {$\cdot$};
	\end{scope}
	\begin{scope}[shift={(2,0)}]
		\draw[dashed] (-0.7,-\yMax) rectangle (0.7,\yMax);
		\draw (-0.5,\yMax) .. controls (-0.5,-0.5+\yMax) and (0,-0.5+\yMax) .. (0,\yMax);
		\draw (-0,-\yMax) .. controls (-0,1-0.5-\yMax) and (0.5,1-0.5-\yMax) .. (0.5,-\yMax);
		\draw (-0.5,-\yMax) .. controls (-0.5,0) and (0.5,0) .. (0.5, \yMax);
		\node (comma) at (1.2,0) {$=$};
	\end{scope}
 \begin{scope}[shift={(4.5,0)}]
	 	\draw[dashed] (-0.7,-2*\yMax) rectangle (0.7,+2*\yMax);
		\begin{scope}[shift={(0,-\yMax)}]
			\draw (0,\yMax) .. controls (0,-0.5+\yMax) and (0.5,-0.5+\yMax) .. (0.5,\yMax);
			\draw (-0.5,-\yMax) .. controls (-0.5,1-0.5-\yMax) and (0,1-0.5-\yMax) .. (0,-\yMax);
			\draw (-0.5,\yMax) .. controls (-0.5,0) and (0.5,0) .. (0.5, -\yMax);
		\end{scope}
		\begin{scope}[shift={(0,+\yMax)}]
			\draw (-0.5,+\yMax) .. controls (-0.5,-0.5+\yMax) and (0,-0.5+\yMax) .. (0,+\yMax);
			\draw (-0,-\yMax) .. controls (-0,1-0.5-\yMax) and (0.5,1-0.5-\yMax) .. (0.5,-\yMax);
			\draw (-0.5,-\yMax) .. controls (-0.5,0) and (0.5,0) .. (0.5, +\yMax);
		\end{scope}
		\node (comma) at (1.2,0) {$=$};
	\end{scope}
 \begin{scope}[shift={(7.3,0)}]
 		\node (q) at (-1,0) {$q$};
		\begin{scope}[shift={(0,0)}]
			\draw[dashed] (-0.7,-\yMax) rectangle (0.7,+\yMax);
			\draw (0.5,\yMax) -- (0.5,-\yMax);
			\draw (-0.5,\yMax) .. controls (-0.5,-0.5+\yMax) and (0,-0.5+\yMax) .. (0,\yMax);
			\draw (-0.5,-\yMax) .. controls (-0.5,1-0.5-\yMax) and (0,1-0.5-\yMax) .. (0,-\yMax);
		\end{scope}
	\end{scope}
\end{tikzpicture}.
\end{align*}
In the last step, we removed the circle and multiplied the result with $q$, then straightened the line and isotoped the box to give an element of $\mathcal{B}_3$. 

It is clear that $TL_3(q)$ is closed under this multiplication.
\end{example}

There is no harm in not drawing the surrounding box, since we know that it is there. But omitting it makes drawing diagrams a lot faster, so we will do just that. We may simply write $TL_n$ for $TL_n(q)$ if $q$ is clear from context. In the following we will fix any $q$.

 The elements in $TL_n$ are called \emph{$n$-tangles}. It is easy to see that the size of $\mathcal{B}_n$, which is the dimension of $TL_n$, is the same as the number of possible ways of putting matching opening and closing brackets on a product of $n+1$ factors, i.e.\ a product of $n$ pairs of factors, like this:
\begin{align*}
a((bc)d), \, a(b(cd)),\, (ab)(cd),\, ((ab)c)d,\, (a(bc))d
\end{align*}
Now, it would surely be nice to have a closed expression for this number, for any $n$. And indeed, this is a well-known problem from combinatorics. The solution is given by the \emph{Catalan numbers},
\begin{align*}
C_n \equiv\frac{1}{n+1} \binom{2n}{n},\quad n\geq 0.
\end{align*}
These numbers grow rather quickly:
\begin{align*}
	1, 1, 2, 5, 14, 42, 132, 429, 1430, 4862, 16796, \ldots
\end{align*}
This makes calculating examples of perfect tangles particularly tedious.

 It is also worthwhile to briefly talk about another description of these algebras, namely the ``classical'' one in terms of generators and relations, see e.g.\ \cite{abramsky_temperley-lieb_2009}. Fix a ring $R$, a nonzero constant $q\in R$, and a natural number $n>1$. The Temperley-Lieb algebra on $n$ strands with loop value $q$ is then the unital associative algebra $A_n(q)$ over $R$ generated by the unit $1_A$ and generators $U_1, U_2,\ldots, U_{n-1}$ subject to the relations
\vspace*{0.5em}
\begin{alignat*}{2}
U_i^2 &= q\cdot U_i &&\\
U_i U_j U_i &= U_i &&\lvert i-j \rvert = 1\\
U_i U_j &= U_j U_i \qquad && \lvert i-j \rvert \geq 2.
\end{alignat*}~

\noindent Translating this into our previous description, we remark that the generators for $n=3$ are precisely the two last elements in $\mathcal{B}_3$ and the straight lines. For example,
\begin{align*}
\begin{tikzpicture}[scale=0.8,baseline]
	\pgfmathsetmacro{\yMax}{0.8};
	\begin{scope}[shift={(0,0)}]
		\draw (-0.5,\yMax) .. controls (-0.5,-0.5+\yMax) and (0,-0.5+\yMax) .. (0,\yMax);
		\draw (-0,-\yMax) .. controls (-0,1-0.5-\yMax) and (0.5,1-0.5-\yMax) .. (0.5,-\yMax);
		\draw (-0.5,-\yMax) .. controls (-0.5,0) and (0.5,0) .. (0.5, \yMax);
		\node (comma) at (1.2,0) {$=$};
	\end{scope}
			\begin{scope}[shift={(2.5,-\yMax)}]
			\draw (-0.5,\yMax) -- (-0.5,-\yMax);
			\draw (0,\yMax) .. controls (0,-0.5+\yMax) and (0.5,-0.5+\yMax) .. (0.5,\yMax);
			\draw (-0,-\yMax) .. controls (-0,1-0.5-\yMax) and (0.5,1-0.5-\yMax) .. (0.5,-\yMax);
		\end{scope}
		\begin{scope}[shift={(2.5,\yMax)}]
			\draw (0.5,\yMax) -- (0.5,-\yMax);
			\draw (-0.5,\yMax) .. controls (-0.5,-0.5+\yMax) and (0,-0.5+\yMax) .. (0,\yMax);
			\draw (-0.5,-\yMax) .. controls (-0.5,1-0.5-\yMax) and (0,1-0.5-\yMax) .. (0,-\yMax);
		\end{scope}
\end{tikzpicture}\,.
\end{align*}
We record the following well-known fact.

\begin{fact}
Both descriptions of the Temperley-Lieb algebra given here are equivalent. In particular it can be shown that all crossingless matchings occurring in $\mathcal{B}_n$ can be obtained by multiplying diagrams consisting only of straight vertical lines, and cups and caps.
\end{fact}

\subsection{Constructing large perfect tangles}\label{sec:FirstGeneralConstruction}
First a quick informal discussion to sharpen the intuition. Consider the following braid.
\begin{align*}
\begin{tikzpicture}[scale = 0.5]
\braid[number of strands = 5] s_4 s_3 s_4 s_2 s_3 s_4 s_1 s_2 s_3 s_4;
\end{tikzpicture}
\end{align*}
Let us try to explain what is happening here: We see that the string starting at position $i$ on the top is attached to position $5-i$ at the bottom. The first rotation will bring the string attached to the left-most point on the top down, and the one attached to the right-most point on the bottom up. But in this setting this is one and the same string! Multiplying with the adjoint of this rotated version will therefore not hinder our efforts in pulling all strings straight. 

Clearly, all subsequent rotations have similar effects. This realization leads to a somewhat educated guess: 
\begin{center}
Tangles built this way are perfect.
\end{center}
The goal of this section is to give both a well-defined construction and a rigorous proof of the statement.

To this end, fix a perfect 2-tangle $T\in P_{2\cdot 2}$, and note that we consider things unshaded.

\begin{definition}\label{def:T}
Let $\widetilde{T}_2\equiv T$, and define inductively for all $n\geq 3$
\begin{align*}
\widetilde{T}_n \equiv\quad
	\begin{tikzpicture}[baseline=-0.5mm]
		\node[draw] (T1) at (0,0) {$\widetilde{T}_{n-1}$};
		\node[draw] (T2) at (2,0) {$T$};
		\path ($(T1.north) + (-0.2,0)$) edge node[left]{$n-2$} ($(T1.north) + (-0.2,0.5)$)
			($(T1.south)$) edge node[left]{$n-1$} ($(T1.south) + (0,-0.5)$)
			($(T2.north) + (-0.2,0)$) edge ($(T2.north) + (-0.2,0.5)$)
			($(T2.north) + (0.2,0)$) edge ($(T2.north) + (0.2,0.5)$)
			($(T2.south) + (0.2,0)$) edge ($(T2.south) + (0.2,-0.5)$);		
		\draw  ($(T1.north) + (0.2,0)$) .. controls (1.2,1.5) and (0.8,-1.5).. ($(T2.south) + (-0.2,0)$);
	\end{tikzpicture}\quad .
\end{align*}

Also, with $\widehat{T}_2 \equiv T$, we define the flipped version of this:
\begin{align*}
\widehat{T}_n \equiv\quad
	\begin{tikzpicture}[baseline=-0.5mm]
		\node[draw] (T1) at (0,0) {$\widehat{T}_{n-1}$};
		\node[draw] (T2) at (2,0) {$T$};
		\path ($(T1.south) + (-0.2,0)$) edge node[left]{$n-2$} ($(T1.south) + (-0.2,-0.5)$)
			($(T1.north)$) edge node[left]{$n-1$} ($(T1.north) + (0,0.5)$)
			($(T2.south) + (-0.2,0)$) edge ($(T2.south) + (-0.2,-0.5)$)
			($(T2.south) + (0.2,0)$) edge ($(T2.south) + (0.2,-0.5)$)
			($(T2.north) + (0.2,0)$) edge ($(T2.north) + (0.2,0.5)$);		
		\draw  ($(T1.south) + (0.2,0)$) .. controls (1.2,-1.5) and (0.8,1.5).. ($(T2.north) + (-0.2,0)$);
	\end{tikzpicture}\quad,
\end{align*}
which is basically the adjoint of the underlying tangle of $\widetilde{T}_n$, with $T$ inserted into all boxes. Both of these are elements in $P_{2n}$.
\end{definition}

\begin{lemma}\label{lem:T=1}
$\widetilde{T}_n \cdot \widetilde{T}_n^* \propto \mathbf{1}_n$, and a similar result holds for $\widehat{T}_n$.
\begin{proof}
By induction. The base case $n=2$ is clear, since $T$ is perfect. Then suppose that it is true for $n\geq 2$. We get
\begin{align*}
	\begin{tikzpicture}[baseline=-1mm]
		\node[draw] (T1) at (0,0.7) {$\widetilde{T}_{n+1}^*$};
		\node[draw] (T2) at (0,-0.7) {$\widetilde{T}_{n+1}$};
		\path
			(T1.north) edge ($(T1.north) + (0,0.5)$)
			(T2.south) edge ($(T2.south) + (0,-.5)$)
			(T1.south) edge (T2.north);
	\end{tikzpicture}
\,&=
	\begin{tikzpicture}[baseline =-8mm]
		\begin{scope}
			\node[draw] (T1) at (0,0) {$\widehat{T}_{n}^*$};
			\node[draw] (T2) at (2,0) {$T^*$};
			\path ($(T1.south) + (-0.2,0)$) edge ($(T1.south) + (-0.2,-0.5)$)
				($(T1.north)$) edge ($(T1.north) + (0,0.5)$)
				($(T2.south) + (-0.2,0)$) edge ($(T2.south) + (-0.2,-0.5)$)
				($(T2.south) + (0.2,0)$) edge ($(T2.south) + (0.2,-0.5)$)
				($(T2.north) + (0.2,0)$) edge ($(T2.north) + (0.2,0.5)$);		
			\draw  ($(T1.south) + (0.2,0)$) .. controls (1.2,-1.5) and (0.8,1.5).. ($(T2.north) + (-0.2,0)$);
		\end{scope}
		\begin{scope}
			\node[draw] (T1) at (0,-1.7) {$\widetilde{T}_{n}$};
			\node[draw] (T2) at (2,-1.5) {$T$};
			\path ($(T1.north) + (-0.2,0)$) edge ($(T1.north) + (-0.2,0.5)$)
				($(T1.south)$) edge ($(T1.south) + (0,-0.5)$)
				($(T2.north) + (-0.2,0)$) edge ($(T2.north) + (-0.2,0.5)$)
				($(T2.north) + (0.2,0)$) edge ($(T2.north) + (0.2,0.5)$)
				($(T2.south) + (0.2,0)$) edge ($(T2.south) + (0.2,-0.5)$);		
			\draw  ($(T1.north) + (0.2,0)$) .. controls (1.2,0) and (0.8,-3).. ($(T2.south) + (-0.2,0)$);
		\end{scope}
	\end{tikzpicture}
\,\propto
	\begin{tikzpicture}[baseline =-8mm]
		\begin{scope}
			\node[draw] (T1) at (0,0) {$\widehat{T}_{n}^*$};
			\path ($(T1.south) + (0,0)$) edge ($(T1.south) + (0,-0.5)$)
				($(T1.north)$) edge ($(T1.north) + (0,0.5)$)
				($(T1.north) + (1,0.5)$) edge ($(T1.south) + (1,-0.5)$);
		\end{scope}
		\begin{scope}
			\node[draw] (T1) at (0,-1.7) {$\widetilde{T}_{n}$};
			\path ($(T1.north) + (0,0)$) edge ($(T1.north) + (0,0.5)$)
				($(T1.south)$) edge ($(T1.south) + (0,-0.5)$)
				($(T1.north) + (1,0.5)$) edge ($(T1.south) + (1,-0.5)$);
		\end{scope}
	\end{tikzpicture}
\, \stackrel{\mathrm{I.H.}}{\propto} \mathbf{1}_{n+1}
\end{align*}
\end{proof}
\end{lemma}

Also note the following:
\begin{equation}\label{eq:concatenation of T}
\begin{tikzpicture}[baseline=-0.5mm]
	\node[draw] (T1) at (0,0) {$\widetilde{T}_{n}$};
	\node[draw] (T2) at (2,0) {$\widetilde{T}_m$};
	\path ($(T1.north) + (-.2,0)$) edge ($(T1.north) + (-0.2,0.5)$)
		($(T1.south)$) edge ($(T1.south) + (0,-0.5)$)
		($(T2.north) + (0,0)$) edge ($(T2.north) + (0,0.5)$)
		($(T2.south) + (0.2,0)$) edge ($(T2.south) + (0.2,-0.5)$);		
	\draw  ($(T1.north) + (0.2,0)$) .. controls (1.2,1.5) and (0.8,-1.5).. ($(T2.south) + (-0.2,0)$);
	\node (eq) at (3.5,0) {$= \widetilde{T}_{n+m-1}$};
\end{tikzpicture}
\end{equation}
(here we are contracting a single strand) and similarly for $\widehat{T}_n$, for all $n,m \geq 2$. 

Consider now the following construction:
\begin{definition}\label{def:B}
Let $A_2 \equiv T = \widetilde{T}_2$, and define for $n\geq 3$
\begin{align*}
A_n \equiv \quad
	\begin{tikzpicture}[baseline = -1mm]
		\node[draw] (T) at (0,0.7) {$\widetilde{T}_{n}$};
		\node[draw] (B) at (0.5, -.7) {$A_{n-1}$};
		\path (T.north) edge ($(T.north) + (0,0.5)$) 
			($(T.south) + (-0.2, 0)$) edge ($(T.south |- B.south) + (-0.2, -0.5)$)
			($(T.south) + (0.2,0)$) edge node[right]{$n-1$} (B.north)
			($(B.south)$) edge ($(B.south |- B.south) + (-0.0, -0.5)$);
	\end{tikzpicture}\quad .
\end{align*}
\end{definition}

Our goal is to show that each $A_n\in P_{2n}$ is perfect. From {Lemma \ref{lem:T=1}} it is clear that $A_n \cdot A_n^* \propto \mathbf{1}_n$. This is because \text{Lemma \ref{lem:T=1}} asserts that $\widetilde{T}$s occurring together with their adjoints give (something proportional to) identities until only one box is left (and some strings going straight from the top to the bottom), namely $A_2 = T$. But $T$ is perfect, so this gives (something proportional to) the identity.

\bigskip\noindent A basic yet important and beautiful property of the $A_n$ is now stated.
\begin{lemma}\label{lem:Bcommuting}
For all $n\geq 3$
\begin{align*}
	A_n = \quad
	\begin{tikzpicture}[baseline = -1mm]
		\node[draw] (T) at (0,-.7) {$\widehat{T}_{n}$};
		\node[draw] (B) at (0.5, .7) {$A_{n-1}$};
		\path (T.south) edge ($(T.south) + (0,-.5)$) 
			($(T.north) + (-0.2, 0)$) edge ($(T.north |- B.north) + (-0.2, 0.5)$)
			($(T.north) + (0.2,0)$) edge (B.south)
			($(B.north)$) edge ($(B.north |- B.north) + (-0.0, 0.5)$);
	\end{tikzpicture}\quad .
\end{align*}
\begin{proof} By induction. For the base case $n=3$, first remember $A_2=\widetilde{T}_2=\widehat{T}_2 = T$. Then
\begin{align*}
A_3 =\,
	\begin{tikzpicture}[baseline=-2mm]
		\node[draw] (T1) at (0,0.5) {$\widetilde{T}_3$};
		\node[draw] (T2) at (0.2,-0.5) {$T$};
		\foreach \x in {-1,0,1}
			\draw ($(T1.north) + (\x *0.2,0)$) edge ($(T1.north) + (\x *0.2,0.4)$);
		\foreach \x in {-0.1,0.1}{
			\draw ($(T1.south) + (0.1+\x,0)$) -- ($(T2.north) + (\x,0)$);
			\draw ($(T2.south) + (\x,0)$) edge ($(T2.south) + (\x,-0.4)$);
		}
		\draw ($(T1.south) + (-0.2,0)$) edge ($(T1.south |- T2.south) + (-0.2,-0.4)$);
	\end{tikzpicture}
\,&=\,
	\begin{tikzpicture}[baseline=-2mm]
		\node[draw] (T1) at (0,0.5) {$T$};
		\node[draw] (T2) at (0.6,-0.5) {$T$};
		\node[draw] (T3) at (1,0.5) {$T$};
		\foreach \x in {-0.1,0.1}{
			\draw ($(T3.north) + (\x,0)$) edge ($(T3.north) + (\x,0.4)$);
			\draw ($(T2.south) + (\x,0)$) edge ($(T2.south) + (\x,-0.4)$);
		}
		\draw 
			($(T1.north) + (-0.1,0)$) edge ($(T1.north) + (-0.1,0.4)$)
			($(T1.south) + (-0.1,0)$) edge ($(T1.south |- T2.south) + (-0.1,-0.4)$);
		\path 
			($(T1.south) + (+0.1,0)$) edge ($(T2.north) + (-0.1,0)$)
			($(T3.south) + (+0.1,0)$) edge ($(T2.north) + (0.1,0)$);
		\draw ($(T1.north) + (+0.1,0)$) .. controls ($(T1.north) + (0.5,0.5)$) and ($(T3.south) - (0.5,0.5)$) .. ($(T3.south) + (-0.1,0)$);
	\end{tikzpicture}
\,=\,
	\begin{tikzpicture}[baseline=-2mm]
		\node[draw] (T1) at (0,-0.5) {$T$};
		\node[draw] (T2) at (1,-0.5) {$T$};
		\node[draw] (T3) at (.6,0.5) {$T$};
		\foreach \x in {-0.1,0.1}{
			\draw ($(T3.north) + (\x,0)$) edge ($(T3.north) + (\x,0.4)$);
			\draw ($(T2.south) + (\x,0)$) edge ($(T2.south) + (\x,-0.4)$);
		}
		\draw 
			($(T1.south) + (-0.1,0)$) edge ($(T1.south) + (-0.1,-0.4)$)
			($(T1.north) + (-0.1,0)$) edge ($(T1.north |- T3.north) + (-0.1,0.4)$);
		\path 
			($(T1.north) + (+0.1,0)$) edge ($(T3.south) + (-0.1,0)$)
			($(T3.south) + (+0.1,0)$) edge ($(T2.north) + (0.1,0)$);
		\draw ($(T1.south) + (+0.1,0)$) .. controls ($(T1.south) + (0.5,-0.5)$) and ($(T2.north) - (0.5,-0.5)$) .. ($(T2.north) + (-0.1,0)$);
	\end{tikzpicture}
\,=\,
	\begin{tikzpicture}[baseline=-2mm]
		\node[draw] (T1) at (0,-0.5) {$\widehat{T}_3$};
		\node[draw] (T2) at (0.2,0.5) {$T$};
		\foreach \x in {-1,0,1}
			\draw ($(T1.south) + (\x *0.2,0)$) edge ($(T1.south) + (\x *0.2,-0.4)$);
		\foreach \x in {-0.1,0.1}{
			\draw ($(T1.north) + (0.1+\x,0)$) -- ($(T2.south) + (\x,0)$);
			\draw ($(T2.north) + (\x,0)$) edge ($(T2.north) + (\x,0.4)$);
		}
		\draw ($(T1.north) + (-0.2,0)$) edge ($(T1.north |- T2.north) + (-0.2,0.4)$);
	\end{tikzpicture}\quad,
\end{align*}
as required. Now suppose that the claim is true up to $n-1$. Then we have
\begin{align*}
A_n &=\,
	\begin{tikzpicture}[baseline=-2mm]
		\node[draw] (T1) at (0,0.8) {$\widetilde{T}_n$};
		\node[draw] (T2) at (0.6,-0.8) {$A_{n-1}$};
		\draw ($(T1.north) + (0 *0.2,0)$) edge ($(T1.north) + (0 *0.2,0.4)$);
		\draw ($(T1.south) + (0.1+0,0)$) -- ($(T2.north) + (0,0)$);
		\draw ($(T2.south) + (0,0)$) edge ($(T2.south) + (0,-0.4)$);
		\draw ($(T1.south) + (-0.2,0)$) edge ($(T1.south |- T2.south) + (-0.2,-0.4)$);
	\end{tikzpicture}
\,=\,
	\begin{tikzpicture}[baseline=-2mm]
		\node[draw] (T1) at (0,0) {$T$};
		\node[draw] (Tt) at (1,1) {$\widetilde{T}_{n-1}$};
		\node[draw] (B) at (1,-1) {$A_{n-1}$};
		\draw ($(T1.south) + (-0.1,0)$)	edge ($(T1.south |- B.south) + (-0.1,-0.4)$);
		\draw ($(T1.north) + (-0.1,0)$)	edge ($(T1.north |- Tt.north) + (-0.1,0.4)$);
		\path
			($(T1.north) + (0.1,0)$) edge ($(Tt.south) + (-0.3,0)$)
			($(T1.south) + (0.1,0)$) edge ($(B.north) + (-0.3,0)$)
			($(B.north) + (+0.3,0)$)edge ($(Tt.south) + (0.3,0)$);
		\path
			(B.south) edge ($(B.south) + (-0,-0.4)$)
			(Tt.north) edge ($(Tt.north) + (-0,0.4)$);
	\end{tikzpicture}
\,\stackrel{\mathrm{I.H.}}{=}\,
	\begin{tikzpicture}[baseline=-2mm]
		\node[draw] (T1) at (0,0) {$T$};
		\node[draw] (Tt) at (1,1) {$\widetilde{T}_{n-1}$};
		\node[draw] (That) at (1,-1) {$\widehat{T}_{n-1}$};
		\node[draw] (B) at (1.3,0) {$A_{n-2}$};
		\draw ($(T1.south) + (-0.1,0)$)	edge ($(T1.south |- That.south) + (-0.1,-0.4)$);
		\draw ($(T1.north) + (-0.1,0)$)	edge ($(T1.north |- Tt.north) + (-0.1,0.4)$);
		\path
			($(T1.north) + (0.1,0)$) edge ($(Tt.south) + (-0.3,0)$)
			($(T1.south) + (0.1,0)$) edge ($(That.north) + (-0.3,0)$)
			($(That.north) + (+0.3,0)$)edge ($(B.south) + (0.0,0)$)
			($(B.north) + (+0.0,0)$)edge ($(Tt.south) + (0.3,0)$);
		\path
			(That.south) edge ($(That.south) + (-0,-0.4)$)
			(Tt.north) edge ($(Tt.north) + (-0,0.4)$);
	\end{tikzpicture}\\[1em]
&=\,
	\begin{tikzpicture}[baseline=-2mm]
		\node[draw] (T1) at (0,0) {$T$};
		\node[draw] (Tt) at (1,1) {$A_{n-1}$};
		\node[draw] (B) at (1,-1) {$\widehat{T}_{n-1}$};
		\draw ($(T1.south) + (-0.1,0)$)	edge ($(T1.south |- B.south) + (-0.1,-0.4)$);
		\draw ($(T1.north) + (-0.1,0)$)	edge ($(T1.north |- Tt.north) + (-0.1,0.4)$);
		\path
			($(T1.north) + (0.1,0)$) edge ($(Tt.south) + (-0.3,0)$)
			($(T1.south) + (0.1,0)$) edge ($(B.north) + (-0.3,0)$)
			($(B.north) + (+0.3,0)$)edge ($(Tt.south) + (0.3,0)$);
		\path
			(B.south) edge ($(B.south) + (-0,-0.4)$)
			(Tt.north) edge ($(Tt.north) + (-0,0.4)$);
	\end{tikzpicture}
\,=\,
	\begin{tikzpicture}[baseline = -2mm]
		\node[draw] (T) at (0,-.7) {$\widehat{T}_{n}$};
		\node[draw] (B) at (0.5, .7) {$A_{n-1}$};
		\path (T.south) edge ($(T.south) + (0,-.5)$) 
			($(T.north) + (-0.2, 0)$) edge ($(T.north |- B.north) + (-0.2, 0.5)$)
			($(T.north) + (0.2,0)$) edge (B.south)
			($(B.north)$) edge ($(B.north |- B.north) + (-0.0, 0.5)$);
	\end{tikzpicture}
\end{align*}
\end{proof}
\end{lemma}

This lemma allows us to switch between two ways of writing $A_n$, which ever way is more comfortable to work with in a given situation. The statement $A_n^*\cdot A_n\propto\mathbf{1}_n$, for example, is now a simple corollary of {Lemma \ref{lem:Bcommuting}} and the remark made after {Definition \ref{def:B}}.

One thing we should note is that if we interpret $\widehat{\phantom{T}}_n$ and $\widetilde{\phantom{T}}_n$ as $n$-tangles with $n-1$ internal boxes, then
\begin{align*}
\widetilde{T}_n^* = \widehat{T^*}_n,
\end{align*}
so that $\widetilde{T}_n\cdot \widehat{T^*}_n\propto\mathbf{1}_n$ by {Lemma \ref{lem:T=1}}.

To prove that the rotations of $A_n$ are also `unitary', we still need to record a few more facts, like the following lemma.
\begin{lemma}\label{lem:half rotated  TT}
Let $1\leq l < n-1$, $n\geq 3$. Then
\begin{align*}
	\begin{tikzpicture}[baseline=-1mm]
			\node[draw] (Tdag) at (0,.7) {$\widetilde{T}_n^*$};
			\node[draw] (T) at (0,-.7) {$\widetilde{T}_n$};
			\path (Tdag.north) edge ($(Tdag.north) + (0,0.5)$)
				(T.south) edge ($(T.south) + (0,-.5)$)
				($(T.north) + (0.2,0)$) edge node[right]{$n-l$} ($(Tdag.south) + (0.2,0)$)
				($(T.north) + (-0.6,0)$) edge node[left] {$l$} ($(T.south) + (-0.6,-0.5)$)
				($(Tdag.south) + (-0.6,0)$) edge node[left] {$l$} ($(Tdag.north) + (-0.6,0.5)$);
			\draw ($(T.north) + (-.2,0)$) arc [start angle=0, end angle=180, radius=0.2cm];
			\draw ($(Tdag.south) + (-.2,0)$) arc [start angle=0, end angle=-180, radius=0.2cm];
	\end{tikzpicture}
\, \propto \,
	\begin{tikzpicture}[baseline=-1mm]
			\node[draw] (Tdag) at (0,.7) {$\widetilde{T}_{l+1}^*$};
			\node[draw] (T) at (0,-.7) {$\widetilde{T}_{l+1}$};
			\path (Tdag.north) edge ($(Tdag.north) + (0,0.5)$)
				(T.south) edge ($(T.south) + (0,-.5)$)
				($(T.north) + (0.2,0)$) edge node[right]{$1$} ($(Tdag.south) + (0.2,0)$)
				($(T.north) + (-0.6,0)$) edge node[left] {$l$} ($(T.south) + (-0.6,-0.5)$)
				($(Tdag.south) + (-0.6,0)$) edge node[left] {$l$} ($(Tdag.north) + (-0.6,0.5)$);
			\draw ($(T.north) + (-.2,0)$) arc [start angle=0, end angle=180, radius=0.2cm];
			\draw ($(Tdag.south) + (-.2,0)$) arc [start angle=0, end angle=-180, radius=0.2cm];
			\draw ($(T.south) + (,-0.5)$) edge node[right]{$n-(l+1)$} ($(Tdag.north) + (1,0.5)$);
	\end{tikzpicture}\,,
\end{align*}
where a string with a 0 next to it is interpreted as no string at all.
\begin{proof}
A computation with careful consideration of the number of strings, and a trick involving \eqref{eq:concatenation of T}, shows the result. We will not write the string count next to each string, so you should definitely do that to help you understand what's going on.

Fix any $n\geq 3$. Here is the calculation:
\begin{align*}
	\begin{tikzpicture}[baseline=-1mm]
			\node[draw] (Tdag) at (0,.7) {$\widetilde{T}_n^*$};
			\node[draw] (T) at (0,-.7) {$\widetilde{T}_n$};
			\path (Tdag.north) edge ($(Tdag.north) + (0,0.5)$)
				(T.south) edge ($(T.south) + (0,-.5)$)
				($(T.north) + (0.2,0)$) edge node[right]{$n-l$} ($(Tdag.south) + (0.2,0)$)
				($(T.north) + (-0.6,0)$) edge node[left] {$l$} ($(T.south) + (-0.6,-0.5)$)
				($(Tdag.south) + (-0.6,0)$) edge node[left] {$l$} ($(Tdag.north) + (-0.6,0.5)$);
			\draw ($(T.north) + (-.2,0)$) arc [start angle=0, end angle=180, radius=0.2cm];
			\draw ($(Tdag.south) + (-.2,0)$) arc [start angle=0, end angle=-180, radius=0.2cm];
	\end{tikzpicture}
\, &= \,
	\begin{tikzpicture}[baseline=-1mm]
		\coordinate (yValue) at (0,1);
		\coordinate (xValue) at (0.7,0);
		\node[draw] (LD) at ($(yValue)-(xValue)$) {$\widetilde{T}_{l+1}^*$};
		\node[draw] (ND) at ($(yValue)+(xValue)$) {$\widetilde{T}_{n-l}^*$};
		\node[draw] (L) at ($(0,0)-(xValue)-(yValue)$) {$\widetilde{T}_{l+1}$};
		\node[draw] (N) at ($(xValue)-(yValue)$) {$\widetilde{T}_{n-l}$};
		\coordinate (yMin) at ($(L.south) + (0,-0.4)$);
		\coordinate (yMax) at ($(LD.north) + (0,0.4)$);
		\foreach \x in {-0.3 -0.4, 0.3}
			\draw ($(L.north) + (\x,0)$) arc[start angle=180, end angle=0, radius=2mm];
		\draw ($(LD.south) + (0.3+0.4,0)$) -- ($(ND.north) + (-0.3-0.4,0)$) arc[start angle=180, end angle=0, radius=2mm];
		\foreach \x in {-0.3 -0.4, 0.3}
			\draw ($(LD.south) + (\x,0)$) arc[start angle=-180, end angle=0, radius=2mm];
		\draw ($(L.north) + (0.3+0.4,0)$) -- ($(N.south) + (-0.3-0.4,0)$) arc[start angle=-180, end angle=0, radius=2mm];
		\foreach \x in {LD.north, ND.north}
			\draw (\x) -- (\x |- yMax);
		\draw ($(LD.south)+(-0.3-0.4,0)$) -- ($(LD.south |- yMax)+(-0.3-0.4,0)$);
		\foreach \x in {L.south, N.south}
			\draw (\x) -- (\x |- yMin);
		\draw ($(L.north)+(-0.3-0.4,0)$) -- ($(L.north |- yMin)+(-0.3-0.4,0)$);
		\draw (ND.south) -- (N.north);
	\end{tikzpicture}
\, \propto \,
	\begin{tikzpicture}[baseline=-1mm]
		\coordinate (yValue) at (0,1);
		\coordinate (xValue) at (0.7,0);
		\node[draw] (LD) at ($(yValue)-(xValue)$) {$\widetilde{T}_{l+1}^*$};
		\node[draw] (L) at ($(0,0)-(xValue)-(yValue)$) {$\widetilde{T}_{l+1}$};
		\coordinate (yMin) at ($(L.south) + (0,-0.4)$);
		\coordinate (yMax) at ($(LD.north) + (0,0.4)$);
		\foreach \x in {-0.3 -0.4}
			\draw ($(L.north) + (\x,0)$) arc[start angle=180, end angle=0, radius=2mm];
		\foreach \x in {-0.3 -0.4}
			\draw ($(LD.south) + (\x,0)$) arc[start angle=-180, end angle=0, radius=2mm];
		\foreach \x in {LD.north}
			\draw (\x) -- (\x |- yMax);
		\draw ($(LD.south)+(-0.3-0.4,0)$) -- ($(LD.south |- yMax)+(-0.3-0.4,0)$);
		\foreach \x in {L.south}
			\draw (\x) -- (\x |- yMin);
		\draw ($(L.north)+(-0.3-0.4,0)$) -- ($(L.north |- yMin)+(-0.3-0.4,0)$);
		\draw ($(LD.south)+(0.3,0)$) -- ($(L.north)+(0.3,0)$);
		\draw ($(yMin) + (1.2,0)$) -- node[right]{$n-l-1$} ($(yMax) + (1.2,0)$);
	\end{tikzpicture}
\end{align*}
The first equality is the trick, the second (well, not equality but proportionality) uses {Lemma \ref{lem:T=1}}.
\end{proof}
\end{lemma}

If you read the previous lemma carefully you will probably have noticed that we excluded the $l=n-1$ case. One reason for this is that we can postpone that part until we are in our main theorem, where we solve it by invoking {Lemma \ref{lem:Bcommuting}}. The other reason is that the formula would be ill-defined, since $n-(n-1) = 1$, but $\widetilde{T}_1$ is undefined.

The following proposition is straightforward, but the result is very useful in one of the last steps in the proof of the main theorem.
\begin{proposition}\label{prop:T rotated n-1 times}
For all $n\geq 2$
\begin{align*}
	\rot^{n-1}\widetilde{T}_n\cdot \left(\rot^{n-1}\widetilde{T}_n\right)^* = 
	\begin{tikzpicture}[baseline=-1mm]
			\node[draw] (Tdag) at (0,.7) {$\widetilde{T}_n^*$};
			\node[draw] (T) at (0,-.7) {$\widetilde{T}_n$};
			\path
				($(Tdag.north) + (-0.2,0)$) edge ($(Tdag.north) + (-0.2,0.5)$)
				($(T.south) + (-0.2,0)$) edge ($(T.south) + (-0.2,-0.5)$)
				($(T.north) + (0.2,0)$) edge ($(Tdag.south) + (0.2,0)$)
				($(T.north) + (-0.6,0)$) edge node[left] {$n-1$} ($(T.south) + (-0.6,-0.5)$)
				($(Tdag.south) + (-0.6,0)$) edge node[left] {$n-1$} ($(Tdag.north) + (-0.6,0.5)$);
			\draw ($(T.north) + (-.2,0)$) arc [start angle=0, end angle=180, radius=0.2cm];
			\draw ($(T.south) + (.2,0)$) arc [start angle=-180, end angle=0, radius=0.2cm] -- node[right]{$n-1$}
				 ($(Tdag.north) + (.6,0)$) arc [start angle=0, end angle=180, radius=0.2cm];
			\draw ($(Tdag.south) + (-.2,0)$) arc [start angle=0, end angle=-180, radius=0.2cm];
	\end{tikzpicture}
\, \propto \mathbf{1}_n
\end{align*}
\end{proposition}
\begin{proof}Omitted. An exercise in mathematical induction for the reader.
\end{proof}

We are now ready to prove our main result.
\begin{theorem}\label{thm:general_existence}
For all $A_n$ the following are true.
\begin{itemize}
	\item[\emph{\text{(i)}}] For all $1\leq l < n$ we have
		\begin{align*}
		\mathrm{rot}^l A_n \cdot \mathrm{rot}^{-l} A_n^* \propto \mathbf{1}_n.
		\end{align*}
	\item[\emph{\text{(ii)}}] $A_n$ is perfect.
\end{itemize}
\begin{proof}
We will first show (i), then (i)$\implies$(ii). The second part is straigthforward.

For (i), first note that this is trivially true for $A_2$, and that proving it for $A_3$ is a basic computation. The base case is thus covered. Now suppose it is true for $n-1$, and consider first $l< n-1$. Then
\begin{align*}
\mathrm{rot}^l A_n \cdot \mathrm{rot}^{-l} A_n^*
&=\,
	\begin{tikzpicture}[baseline=-1mm]
		\node[draw] (BD) at (1,2) {$A_{n-1}^*$};
		\node[draw] (TD) at (0,0.8) {$\widetilde{T}_n^*$};
		\node[draw] (T) at (0,-0.8) {$\widetilde{T}_n$};
		\node[draw] (B) at (1, -2) {$A_{n-1}$};
		\coordinate (yMax) at ($(BD.north)+(0,0.4)$);
		\coordinate (yMin) at ($(B.south)+(0,-0.4)$);
		\draw 
			($(B.south)+(0.4,0)$) arc[start angle=-180, end angle=0, radius=2mm] -- node[right] {$l$}
			($(BD.north)+(0.8,0)$) arc[start angle=0, end angle=180, radius=2mm];
		\path 
			($(TD.north)+(0.2,0)$) edge ($(BD.south)+(-0.4,0)$)
			($(TD.south)+(0.2,0)$) edgenode[right] {$n-l$}  ($(T.north)+(0.2,0)$)
			($(T.south)+(0.2,0)$) edge ($(B.north)+(-0.4,0)$);
		\draw
			($(TD.south)+(-0.2,0)$) arc[start angle=0, end angle=-180, radius=2mm] edge node[left] {$l$} ($(TD.south |- yMax)+(-0.6,0)$);
		\draw
			($(T.north)+(-0.2,0)$) arc[start angle=0, end angle=180, radius=2mm] edge node[left] {$l$} ($(TD.south |- yMin)+(-0.6,0)$);
		\path
			($(BD.north)+(-0.2,0)$) edge ($(BD.north |- yMax)+(-0.2,0)$)
			($(TD.north)+(-0.2,0)$) edge ($(TD.north |- yMax)+(-0.2,0)$)
			($(T.south)+(-0.2,0)$) edge ($(T.south |- yMin)+(-0.2,0)$)
			($(B.south)+(-0.2,0)$) edge ($(B.south |- yMin)+(-0.2,0)$);
	\end{tikzpicture}
\,=\,
	\begin{tikzpicture}[baseline=-1mm]
		\node[draw] (BD) at (1,2) {$A_{n-1}^*$};
		\node[draw] (TD) at (0,0.8) {$\widetilde{T}_{l+1}^*$};
		\node[draw] (T) at (0,-0.8) {$\widetilde{T}_{l+1}$};
		\node[draw] (B) at (1, -2) {$A_{n-1}$};
		\coordinate (yMax) at ($(BD.north)+(0,0.4)$);
		\coordinate (yMin) at ($(B.south)+(0,-0.4)$);
		\draw 
			($(B.south)+(0.4,0)$) arc[start angle=-180, end angle=0, radius=2mm] -- node[right] {$l$}
			($(BD.north)+(0.8,0)$) arc[start angle=0, end angle=180, radius=2mm];
		\path 
			($(TD.north)+(0.2,0)$) edge node[right]{$l$} ($(BD.south)+(-0.4,0)$)
			($(TD.south)+(0.2,0)$) edge node[right] {$1$}  ($(T.north)+(0.2,0)$)
			($(T.south)+(0.2,0)$) edge ($(B.north)+(-0.4,0)$);
		\draw
			($(BD.south)+(+0.4,0)$) edge node[above, rotate=90]{{\small$n-1-l$}}  ($(B.north)+(0.4,0)$);
		\draw
			($(TD.south)+(-0.2,0)$) arc[start angle=0, end angle=-180, radius=2mm] edge node[left] {$l$} ($(TD.south |- yMax)+(-0.6,0)$);
		\draw
			($(T.north)+(-0.2,0)$) arc[start angle=0, end angle=180, radius=2mm] edge node[left] {$l$} ($(TD.south |- yMin)+(-0.6,0)$);
		\path
			($(BD.north)+(-0.2,0)$) edge ($(BD.north |- yMax)+(-0.2,0)$)
			($(TD.north)+(-0.2,0)$) edge ($(TD.north |- yMax)+(-0.2,0)$)
			($(T.south)+(-0.2,0)$) edge ($(T.south |- yMin)+(-0.2,0)$)
			($(B.south)+(-0.2,0)$) edge ($(B.south |- yMin)+(-0.2,0)$);
	\end{tikzpicture}\\[2em]
\, &\propto \,
	\begin{tikzpicture}[baseline=-1mm]
			\node[draw] (Tdag) at (0,.7) {$\widetilde{T}_{l+1}^*$};
			\node[draw] (T) at (0,-.7) {$\widetilde{T}_{l+1}$};
			\path
				($(Tdag.north) + (-0.2,0)$) edge ($(Tdag.north) + (-0.2,0.5)$)
				($(T.south) + (-0.2,0)$) edge ($(T.south) + (-0.2,-0.5)$)
				($(T.north) + (0.2,0)$) edge ($(Tdag.south) + (0.2,0)$)
				($(T.north) + (-0.6,0)$) edge node[left] {$l$} ($(T.south) + (-0.6,-0.5)$)
				($(Tdag.south) + (-0.6,0)$) edge node[left] {$l$} ($(Tdag.north) + (-0.6,0.5)$);
			\draw ($(Tdag.north) + (1.5,0.5)$) -- node[right]{$n-1-l$} ($(T.south) + (1.5,-0.5)$);
			\draw ($(T.north) + (-.2,0)$) arc [start angle=0, end angle=180, radius=0.2cm];
			\draw ($(T.south) + (.2,0)$) arc [start angle=-180, end angle=0, radius=0.2cm] -- node[right]{$l$}
				 ($(Tdag.north) + (.6,0)$) arc [start angle=0, end angle=180, radius=0.2cm];
			\draw ($(Tdag.south) + (-.2,0)$) arc [start angle=0, end angle=-180, radius=0.2cm];
	\end{tikzpicture}\\[2em]
&\propto \mathbf{1}_n
\end{align*}
What's happening here is the following. The first step is writing out the definition. The second is applying {Lemma \ref{lem:half rotated  TT}}, after which we are at the induction step: it's true that $\mathrm{rot}^l A_{n-1} \cdot \mathrm{rot}^{-l} A_{n-1}^* \propto \mathbf{1}_{n-1}$. Note that the largest value of $l$ here is $(n-1)-1$.
Last but not least, we invoke {Proposition \ref{prop:T rotated n-1 times}}, which gives the desired result.

For $l=n-1$ note that the rotation bends \emph{all} legs of $A_{n-1}$ up- and all legs of $A_{n-1}^*$ downwards, thereby contracting the $A$s along half of their legs. That is we have the product $A_{n-1}^*\cdot  A_{n-1}$, which, as we have argued before, is proportional to $\mathbf{1}_{n-1}$. Then all that's left is applying {Proposition \ref{prop:T rotated n-1 times}} one time and the proof of part (i) is finished.

\bigskip
Part (ii) of the theorem follows from 
\begin{align*}
0\leq l <k : \mathrm{rot}^{k+l} T\cdot\mathrm{rot}^{-(k+l)}T^* \propto \mathbf{1}_k \quad\iff\quad \mathrm{rot}^{-l}T^* \cdot \mathrm{rot}^{l} T\propto \mathbf{1}_k,
\end{align*}
and an argument about the validity of `horizontally flipped' versions of {Lemma \ref{lem:half rotated  TT}} and {Proposition \ref{prop:T rotated n-1 times}}, after seeing that the base case is (almost trivially) true.
\end{proof}
\end{theorem}

The perfect tangles constructed in this section are certainly very pretty. In \figref{fig:examples_general_construction} we see a few examples which are in $P_{2n}$ for $n=3,4,5,6$, respectively. The perfect tangle $T\in P_{2\cdot 2}$ is represented as a vertex and should be interpreted as being in standard form.
\begin{figure}[!htp]\centering
\begin{tikzpicture}[inner sep=0.5mm,scale=1]
	\foreach \order in {1,...,4}{
		\begin{scope}[xshift=1/2*\order*\order cm]
			\coordinate (yMax) at (0, .7+ 0.5*\order);
			\coordinate (yMin) at (0,-.7-.5*\order);
			\foreach \x in {0,...,\order}{
				\foreach \y in {0,...,\x}{		
					\node[draw, fill,circle] (x\x) at (.8*\x,.5*\x -1*\y) {};
					\ifthenelse{\NOT \x = \order}{
						\draw (x\x) -- ($(x\x) + (.8, .5 - 1)$);
						\draw (x\x) -- ($(x\x) + (.8, -.5 + 1)$);
						\ifthenelse{\y=\x}{
							\ifthenelse{\y=0}{\draw (x\x) -- (x\x |- yMax);}{}
							\draw (x\x) -- (x\x |- yMin);}{\ifthenelse{\y=0}{\draw (x\x) -- (x\x |- yMax);}{}
						}
					}{
						\ifthenelse{\y=0}{
							\draw (x\x) -- ($(x\x) + (0,-\order)$);
							\draw ($(x\x.north) + (0.05,-0.05)$) -- ($(x\x.north |- yMax) + (0.05, 0)$);
							\draw ($(x\x.north) - (0.05, 0.05)$) -- ($(x\x.north |- yMax) - (0.05,0)$);
						}{
							\ifthenelse{\y=\x}{
							\draw ($(x\x.south) - (0.05,-0.05)$) -- ($(x\x.south |- yMin) - (0.05, 0)$);
							\draw ($(x\x.south) +(0.05, 0.05)$) -- ($(x\x.south |- yMin) + (0.05,0)$);
							}{}
						}
					}
				}
			}
			\ifthenelse{\NOT \order=4}{\node (comma) at (0.2+.8*\order,0) {$,$};}{}
		\end{scope}
	}
\end{tikzpicture}
\caption[Examples of perfect tangles obtained by the iterative construction proved in {Theorem \ref{thm:general_existence}}]{The perfect tangles $A_3,A_4,A_5,A_6$. Vertices represent the generating perfect 2-tangle $T$ in standard form.}
\label{fig:examples_general_construction}
\end{figure}
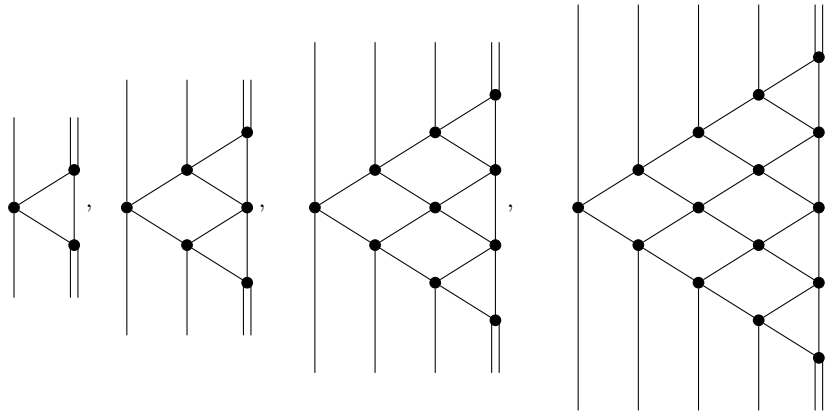

It is clear that what this theorem actually does is that it gives us for each $n$ an $n$-tangle with $\frac{n(n-1)}{2}$ internal 2-boxes, such that if we fix a perfect 2-tangle and insert a copy into each box the resulting element is perfect.


Letting $\mathcal{P}_n$ denote the set of perfect $n$-tangles --- including adjoints and rotations of tangles already in $\mathcal{P}_n$ --- then {Theorem \ref{thm:general_existence}} shows a lower bound on the size of this set:
\begin{align}\label{eq:lower_bound no1}
\lvert \mathcal{P}_2 \rvert \leq \lvert \mathcal{P}_n \rvert \quad \forall n\geq 2.
\end{align}
So whenever $\mathcal{P}_2\neq \emptyset$, the general existence of perfect tangles in spaces indexed by an even color is established.

\subsection{Code}
The code employed to deduce the examples of this Letter can be run using \texttt{sage}, an open source computer algebra system which interfaces many wonderful little helpers, like \emph{Singular}, \emph{GAP}, or various \emph{python} libraries such as \emph{sympy}. \texttt{sage} itself is python-based and comes with a very good documentation, which may readily be found on-line.

Finding the relevant equations can be done using \texttt{sage}, but usually only an additional use of \emph{Mathematica's} powerful \texttt{Reduce} feature allows us to find actual solutions.

We assume the reader is familiar with the basics of programming. The syntax of python and thus \texttt{sage} is very easy to comprehend, the only thing a bit challenging at first might be list comprehension, but that is quickly learned.

The entire (more-or-less documented and therefore self-explanatory) code produced for this Letter can be found on the author's \emph{github} at
\begin{center}
\url{github.com/zerschmetterling/ppt_code}
\end{center}

\end{document}